\documentclass{article}
\usepackage{amsfonts,color}
\usepackage{amsmath,mathabx}
\usepackage{caption}
\usepackage{enumitem}
\usepackage{lscape}
\usepackage{subcaption}
\usepackage{graphicx}
\usepackage{hyperref}
\usepackage{cleveref}
\usepackage{setspace}
\usepackage{tikz}
\usepackage{tabularx}
\usepackage{adjustbox}
\usepackage{marginnote}
\usepackage{tablefootnote} 
\usepackage{url}
\usetikzlibrary{patterns}

\newtheorem{theorem}{Theorem}

\newtheorem{corollary}[theorem]{Corollary}

\newtheorem{definition}[theorem]{Definition}
\newtheorem{example}[theorem]{Example}

\newtheorem{lemma}[theorem]{Lemma}

\newtheorem{proposition}[theorem]{Proposition}

\newtheorem{remark}[theorem]{Remark}

\newenvironment{proof}[1][Proof]{\noindent\textbf{#1.} }{\ \rule{0.5em}{0.5em}}

\newcommand{\R}{\mathbb{R}}
\newcommand{\E}{\operatorname{\mathbb{E}}}
\newcommand{\diff}{\mathrm{d}}

\newcommand{\B}{\mathcal{B}}

\newcommand{\F}{\mathcal{F}}
\newcommand{\M}{\mathcal{M}}

\newcommand{\trho}{\tilde{\rho}}

\newcommand{\tpsi}{\tilde{\psi}}

\DeclareMathOperator*{\argmin}{arg\,min}
\DeclareMathOperator*{\esssup}{ess\,sup}

\DeclareMathOperator*{\Var}{Var}

\newcommand{\keywords}[1]{\textbf{Keywords } #1}

\newcommand{\sqb}[1]{\ensuremath{ \left[ #1 \right] }}
\newcommand{\of}[1]{\ensuremath{\left( #1 \right)}}
\newcommand{\norm}[1]{\ensuremath{ \left\Vert #1 \right\Vert }}
\newcommand{\cb}[1]{\ensuremath{ \left\{ #1 \right\} }}

\let\abs=\envert

\title{Elicitability of Return Risk Measures\thanks{We are very grateful to seminar participants at the University of Vienna and the University of Amsterdam for their comments and suggestions. 
This research was funded in part by the Netherlands Organization for Scientific Research under an NWO-Vici grant 2020--2025 (Ayg\"un and Laeven).}
}
\author{M\"ucahit Ayg\"un 
\\
{\footnotesize Dept. of Quantitative Economics}\\
{\footnotesize University of Amsterdam}\\
{\footnotesize and Tinbergen Institute}\\
{\footnotesize \texttt{M.Aygun@uva.nl}}
\\
\and Fabio Bellini 
\\
{\footnotesize Dept. of Statistics and Quantitative Methods}\\
{\footnotesize University of Milano Bicocca}\\
{\footnotesize \texttt{Fabio.Bellini@unimib.it}}\\
\and Roger J.~A.~Laeven
\\
{\footnotesize Dept. of Quantitative Economics}\\
{\footnotesize University of Amsterdam, CentER}\\
{\footnotesize and EURANDOM}\\
{\footnotesize \texttt{R.J.A.Laeven@uva.nl}}
}

\begin{document}
	
\maketitle
\begin{abstract}
Informally, a risk measure is said to be elicitable if there exists a suitable
scoring function such that minimizing its expected value recovers the risk measure.\
In this paper, we analyze the elicitability properties of
the class of return risk measures (i.e., normalized, monotone and positively homogeneous risk measures).\
First, we provide dual representation results for convex and geometrically convex return risk measures.\
Next, we establish new axiomatic characterizations of Orlicz premia (i.e., Luxemburg norms).\
More specifically, we prove, under different sets of conditions, that Orlicz premia naturally arise
as the only elicitable return risk measures.\
Finally, we provide a general family of strictly consistent scoring functions for Orlicz premia,
a myriad of specific examples
and a mixture representation suitable for constructing Murphy diagrams.\
\end{abstract}


\section{Introduction}

Since the seminal work of Savage (\cite{S71}) and Osband (\cite{O85}), an expanding and increasingly sophisticated literature
has studied elicitability properties of risk measures.\
Classes of risk measures may, or may not, admit families of strictly consistent scoring functions, and hence be elicitable,
with important implications for evaluating model performance and competing forecasts (see e.g.,\ \cite{G11}, \cite{P11}, \cite{BB15}, \cite{Z16}, \cite{NZ17}, \cite{CDS10}, and the references therein).\
For example, Average Value-at-Risk per se is not elicitable, but it is jointly elicitable with Value-at-Risk since it admits bivariate strictly consistent scoring functions (\cite{G11}, \cite{FZ16}).

Recently, \cite{BLR18} introduced the class of return risk measures, consisting of normalized, monotone and positively homogeneous risk measures.\
Return risk measures provide relative (or geometric) assessments of risk.\
They evaluate how much additional riskless log-return makes a financial position acceptable---whence their name.\
They constitute the relative counterparts of the class of monetary risk measures (\cite{FS11}, \cite{D12}), reminiscent of how relative risk aversion relates to absolute risk aversion.\
Their dynamic extensions, dynamic return risk measures, have been studied in
\cite{BLR21}.\footnote{Return risk measures
that allow for probability distortion were recently analyzed in \cite{WX22},
whereas applications of return risk measures to capital allocation can be found in \cite{MFS21} and \cite{CCR21}.}

Whereas elicitability properties of monetary risk measures are by now quite well understood,
little is known about the elicitability properties of return risk measures.\
This paper aims to fill this gap by analyzing the elicitability properties of return risk measures, with a particular emphasis on Orlicz premia, also known as Luxemburg norms, 
which as we will see play a central role in the theory of elicitable return risk measures.\
Orlicz premia, and the links between risk measures and Orlicz space theory, have been extensively studied in the financial and actuarial mathematics literature
(see e.g.,\ \cite{HG82}, \cite{BF08}, \cite{CL08}, \cite{CL09}, \cite{D12}, \cite{LS14}, \cite{BLR18}, \cite{BLR21} and the references therein); however, their connection to statistical decision theory in general, and elicitability in particular, has not been uncovered to our best knowledge.

This paper makes three main contributions.\
We start by providing dual representation results for convex and geometrically convex return risk measures
and clarify their precise relationship.\
In full generality, the dual representation takes the form of a supremum of discounted logarithmic certainty equivalents,
where the discount factor can be interpreted as an index of model plausibility under ambiguity.\
We show that convex return risk measures occur as a special case in the richer class of geometrically convex return risk measures, and
we also analyze their law-invariant representations.\
Furthermore, we introduce and analyze the class of optimized return risk measures and derive their dual representation.

Second, we establish new characterization results for Orlicz premia.\ 
We prove that Orlicz premia naturally arise as the only return risk measures that are elicitable.\
It has been shown in \cite{O85} that an elicitable risk measure must satisfy the convex
level sets (CxLS) property.\
We establish that a law-invariant geometrically convex return risk measure with the CxLS property is necessarily an Orlicz premium.\
We also show that requiring identifiability for return risk measures singles out the class of Orlicz risk measures:
under weak regularity conditions, they are the only identifiable, law-invariant, monotone and positively homogeneous measures of risk.\
These are our central results, the preparations and mathematical details of which are somewhat involved.

Third, we provide a general, rich family of scoring functions that we prove to be strictly consistent with Orlicz premia. 
A plethora of examples illustrates the generality of our new family of scoring functions.\
Special attention is devoted to scoring functions of the relative error form in view of their appealing properties in forecast evaluation.\
We also provide a mixture representation of the general family of scoring functions in terms of elementary scoring functions, depending on a low-dimensional parameter.\ 
This enables the use of so-called Murphy diagrams to compare competing forecasts simultaneously with respect to a full class of strictly consistent scoring functions, and we illustrate this in two examples.\

Statistical decision theory demonstrates that some classes of functionals
do not allow for meaningful point forecast evaluation by means of expected scores.\
Functionals that admit a strictly consistent scoring function,
guaranteeing that accurate forecasts are rewarded more than inaccurate forecasts,
are referred to as elicitable (see Definition~\ref{def:eli} for a formal definition).\
Our characterization results reveal the important place of the class of Orlicz premia
in the extensive literature on risk measures.\
This is graphically illustrated in Figure~\ref{fig:Vdiagram}.\
We know from \cite{W06}, \cite{BB15} and \cite{DBBZ14} that convex shortfall risk measures
occur as the subclass of monetary risk measures that are elicitable.\
Furthermore, the only elicitable law-invariant coherent risk measures are given by expectiles (\cite{Z16}, \cite{DBBZ14}).\
We establish in Theorems~\ref{th:axiom_2} and~\ref{th:ident} that Orlicz premia naturally arise as elicitable return risk measures.\
Furthermore, in Theorem~\ref{th:expectiles}, we provide a direct proof of the result that the only convex Orlicz premia that are translation invariant (and, hence, coherent risk measures) are the expectiles.\

The rest of this paper is organized as follows.\
In Section~\ref{sec:rrm}, we recall the general properties of return risk measures and derive some useful continuity properties.\
In Section~\ref{sec:dual}, we provide dual representation results for geometrically convex 
and convex return risk measures,
explicate their connection and analyze optimized return risk measures.\
In Section~\ref{sec:Orl}, we establish our characterization results for Orlicz premia.\
Section~\ref{sec:scoring} presents our results on families of scoring functions strictly consistent with Orlicz premia 
including many examples.\

\begin{figure}
\begin{tikzpicture}
\begin{scope}
\draw[black,very thick] (-2,0) circle(3);
\draw[black,very thick](2,0) circle(3);
\clip (-2,0) circle(3);
\clip(2,0) circle(3);
\draw[black,thick] (0,-2.7) circle(3);
\draw[black,thick] (0,4.2) circle(3);
\draw[black,thick] (0,-3.3) circle(3);
\fill[pattern=vertical lines, pattern color=black!30!white](0,-2.7) circle(3);
\fill[pattern=horizontal lines, pattern color=black!30!white](0,4.2) circle(3);
\end{scope}
\begin{scope}
\clip (-2,0) circle(3);
\draw[black,thick] (0,-3.8) circle(3);
\fill[pattern=north west lines, pattern color=black!30!white](0,-3.8) circle(3);
\end{scope}
\begin{scope}
\clip (2,0) circle(3);
\draw[black,thick] (0,-3.8) circle(3);
\fill[pattern=north east lines, pattern color=blue!30!white] (0,-3.8) circle(3);
\end{scope}
\begin{scope}
\draw[black,thick] (-2,0) circle(3);
\draw[black,thick](2,0) circle(3);
\clip (-2,0) circle(3);
\clip(2,0) circle(3);
\draw[black,thick] (0,-3.8) circle(3);
\end{scope}
\draw[black,very thick] (-2,0) circle(3);
\draw[blue,very thick](2,0) circle(3);
\draw[black,thick] (0,1.3) circle (0.00000001pt) node[anchor=south]{V@R};
\draw[black,thick] (0,-0.8) circle (0.00000001pt) node[anchor=south]{AV@R};
\filldraw[black] (0,-1) circle (0.000001pt) node[anchor=north]{\footnotesize Expectiles};
\filldraw[black] (0,0.85) circle (0.000001pt) node[anchor=north]{Coherent};
\draw [-to](0,0.45) -- (0,0.32);
\filldraw[blue] (2.5,-1.3) circle (0.000001pt) node[anchor=south]{
Orlicz};
\draw [-to](2.5,-1.3) -- (2,-1.5);
\filldraw[black] (-2.5,-1.3) circle (0.000001pt) node[anchor=south]{
Shortfall};
\draw [-to](-2.5,-1.3) -- (-2,-1.5);
\filldraw[blue] (4,3) circle (0.0000001pt) node[anchor=south]{Return Risk Measures};
\draw [-to](3.8,3) -- (3.5,2.7);
\filldraw[black] (-4,3) circle (0.0000001pt) node[anchor=south]{Monetary Risk Measures};
\draw [-to](-3.8,3) -- (-3.5,2.7);
\end{tikzpicture}
\caption{Venn Diagram of Classes of Risk Measures.}
\label{fig:Vdiagram}
{\small \textit{Notes:} This figure graphically illustrates the relationships between the classes of monetary risk measures and return risk measures and some of their prominent subclasses.
The intersection between monetary and return risk measures includes Value-at-Risk (V@R), the class of coherent risk measures in which Average Value-at-Risk (AV@R) occurs as a special case, and the class of expectiles.
Convex shortfall risk measures occur as the subclass of monetary risk measures that are elicitable.  
The only coherent shortfall risk measures are the expectiles. 
As we establish in this paper, convex Orlicz premia arise as elicitable return risk measures. 
Their intersection with the class of coherent risk measures is again given by the expectiles.}
\end{figure}
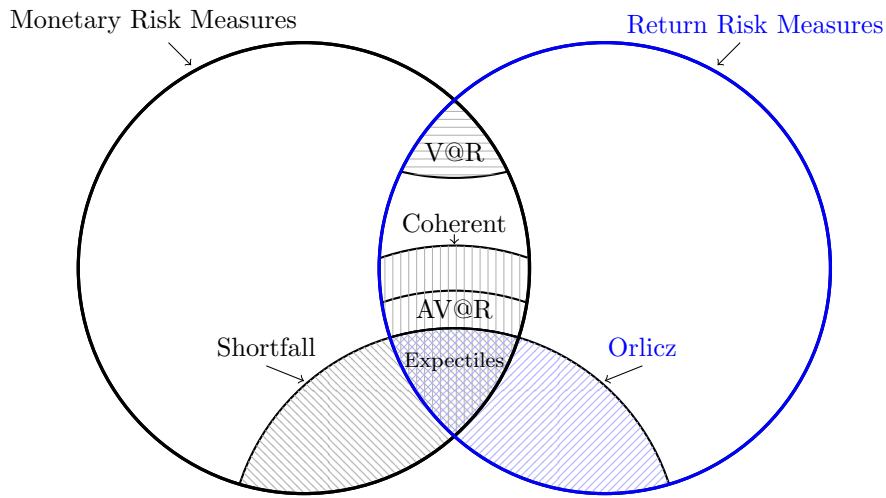
	
\section{Return risk measures}\label{sec:rrm}


Let $(\Omega, \F, P)$ be a nonatomic probability space.\ 
In the present paper, random variables $X \colon \Omega \to \R$  represent financial losses.\
We will consider finite-valued risk measures 
defined on $L^{\infty}(\Omega, \F, P)$ or on its subsets
$L^{\infty}_{+}(\Omega, \F, P):=\{ X \in L^\infty \mid X \geq 0 \; P \mbox{-a.s.} \}$ and
$L^{\infty}_{++}(\Omega, \F, P):=\{ X \in L^\infty \mid X > 0 \; P \mbox{-a.s.} \}$.\
Equalities and inequalities between random variables are meant to hold $P$-almost surely without further explicit mentioning.
\smallskip
\begin{definition}
We say that a functional $\rho \colon L^{\infty}(\Omega, \F, P) \to \R$ is:
\begin{enumerate}[label={\alph*}), left = 0pt, itemsep=0pt]
\item  translation invariant if $\rho(X+h)=\rho(X)+h, \, \forall h \in \R, \, \forall X \in L^{\infty}$
\item  monotone if $X \leq Y \Rightarrow \rho(X) \leq \rho (Y)$
\item  monetary if $\rho$ is translation invariant, monotone and satisfies $\rho(0)=0$
\item positively homogeneous if $\rho( \lambda X)= \lambda \rho(X), \, \forall \lambda \geq 0, \, \forall X \in L^{\infty}$
\item convex if
$
\rho (\alpha X +(1-\alpha)Y) \leq \alpha \rho (X) + (1 - \alpha) \rho (Y), \, \forall X,Y \in L^{\infty}, \, \forall \alpha \in (0,1)$
\item coherent if it is monetary, convex and positively homogeneous
\item law invariant if $X \sim Y \Rightarrow \rho(X)=\rho(Y)$, where $X\sim Y$ means that $X$ and $Y$ have the same distribution.
\end{enumerate}
\end{definition}
A law-invariant functional on $L^\infty(\Omega, \F, P) $ induces a functional on $\M_{1,c}(\R)$, the set of probability measures with compact support in $\R$, by means of
\[
\rho(F) := \rho (X), \text { if } X \sim F,
\]
where each probability measure $\mu\in\M_{1,c}(\R)$ is identified with its distribution function  $F(x):=\mu(-\infty,x]$.

We recall from \cite{BLR18} the notions of return risk measure and of its associated multiplicative acceptance set.
	
\begin{definition}\label{def:ret-rm}
A return risk measure $\trho \colon L^\infty_{+} \to [0,+\infty)$ is a positively homogeneous and monotone functional satisfying $\trho(1)=1$. 
Its corresponding multiplicative acceptance set (at the level of random variables) is $B_{\trho} = \{X \in L^\infty_{+} \mid \trho(X) \leq 1\}.$
\end{definition}

For return risk measures the notion of geometric convexity---also known as 
multiplicative convexity or GG-convexity for functions on the positive real line (see e.g., \cite{N00})---will be of interest in what follows.
	
\begin{definition}
A functional $\trho \colon L^\infty_{+} \to [0,+\infty)$ is geometrically convex if for each $X,Y \in L^{\infty}_{+}$ and  $\alpha \in (0,1)$  it holds that
\[
\trho(X^\alpha Y^{1-\alpha}) \leq \trho^\alpha(X)\trho^{1-\alpha}(Y).
\]
\label{def:georetrm}
\end{definition}

We will show in Lemma~\ref{lemma:convex-GGconvex} that convex return risk measures are also geometrically convex. 
The class of geometrically convex risk measures is strictly larger than the class of convex return risk measures, a nonconvex example of the former being the logarithmic certainty equivalent $\trho(X)=\exp \E [\log X] $.

A one-to-one correspondence between return risk measures and monetary risk measures has been outlined in \cite{BLR18} as follows:\ given a monetary risk measure $\rho \colon  L^{\infty} \to \R $, the associated return risk measure $\trho \colon L^{\infty}_{++} \to (0, +\infty)$ is given by
\begin{equation}
\label{eq:trho-1}
\trho(X):=\exp \left( \rho \left( \log(X) \right) \right),
\end{equation}
and, \textit{vice versa}, given a return risk measure $\trho \colon L^{\infty}_{++} \to (0, +\infty)$, the associated monetary risk measure $\rho \colon  L^{\infty} \to \R $ is
\begin{equation}
\label{eq:trho-2}
\rho(Y):=\log \left( \trho \left( \exp(Y) \right) \right).
\end{equation}
\noindent
The main properties of this correspondence are recalled in the following lemma.
\begin{lemma}  \label{lemma:rho_trho}
Let $\rho \colon  L^{\infty} \to \R $ and $\trho \colon L^{\infty}_{++} \to (0, +\infty)$ be as in \eqref{eq:trho-1} and \eqref{eq:trho-2}.\
Then:
\begin{enumerate}[label={\alph*}), left = 0pt, itemsep=0pt]
\item  $\rho(0)=0 \iff \trho(1)=1$
\item  $\rho$ is translation invariant $\iff$ $\trho$ is positively homogeneous
\item  $\rho$ is monotone $\iff$ $\trho$ is monotone
\item  $\rho$ is convex $\iff$ $\trho$ is geometrically convex
\item  $\rho$ is law invariant $\iff$ $\trho$ is law invariant
\item  if $\rho$ is law invariant, then,
for $F \in \M_{1,c}(0,+\infty)$,
\begin{equation}
\label{eq:trho_dist}
\trho (F) =  \exp \left (\rho \left ( F (e^t) \right ) \right ).
\end{equation}
\end{enumerate}
\end{lemma}


In Section~\ref{sec:dual}, we will establish dual representations for geometrically convex return risk measures.\ 
It will turn out that for return risk measures the definitions of Fatou and Lebesgue properties have to be slightly modified.\ 
We introduce both the usual and the modified versions in the definition below.
\begin{definition}
A risk measure $\rho$ has the Fatou property if
\begin{align*}
X_n \overset{P}{\to} X, \, \norm{X_n}_{\infty} \leq k \implies \rho(X)\leq \liminf_{n\to +\infty} \rho(X_n),
\end{align*}
whereas it has the Lebesgue property if
\begin{align*}
X_n \overset{P}{\to} X, \, \norm{X_n}_{\infty} \leq k \implies \rho(X_n)\to  \rho(X).
\end{align*}
A return risk measure $\trho$ has the lower-bounded Fatou property if
\begin{align*}
X_n \overset{P}{\to} X, \, \norm{X_n}_{\infty} \leq k,   X_n\geq  c>0 \implies \Tilde{\rho}(X)\leq \liminf_{n\to +\infty} \Tilde{\rho}(X_n),
\end{align*}
whereas it has the lower-bounded Lebesgue property if
\begin{align*}
X_n \overset{P}{\to} X, \, \norm{X_n}_{\infty} \leq k,   X_n\geq  c>0 \implies \Tilde{\rho}(X_n)\to  \Tilde{\rho}(X).
\end{align*}
\end{definition}
	
\noindent
Clearly the lower-bounded Lebesgue property implies the lower-bounded Fatou property. 
Both properties are weaker than the usual ones, requiring respectively lower semicontinuity and continuity under more restrictive assumptions.

\begin{lemma}\label{lemma:Fatou}
Let $\rho \colon  L^{\infty} \to \R$ and $\Tilde{\rho} \colon L^{\infty}_{++} \to (0,+\infty)$ be as in \eqref{eq:trho-1} and \eqref{eq:trho-2}. Then:
\begin{enumerate}[left=0pt, itemsep=0pt]
\item[(i)] $\Tilde{\rho}$ has the lower-bounded Fatou property if and only if $\rho$ has the Fatou property
\item[(ii)] $\Tilde{\rho}$ has the lower-bounded Lebesgue property if and only if $\rho$ has the Lebesgue property.
\end{enumerate}
\end{lemma}
\begin{proof}
(i) Let $X_n \in L^\infty_{++}$ satisfy $X_n \overset{P}{\to} X$, $\Vert X_n \Vert_\infty \leq k$, $X_n \geq c >0$. 
By the continuous mapping theorem, it follows that $\log(X_n) \overset{P}{\to} \log(X)$, and $\Vert \log(X_n) \Vert_{\infty} \leq \max(\log k,  -\log c)$, so from the Fatou property of $\rho$ it follows that
\[
\rho(\log X) \leq \liminf_{n \to +\infty} \rho(\log(X_n)),
\]
and exponentiating both sides we get $\trho(X) \leq \liminf_{n \to +\infty} \trho(\log(X_n))$.\ 
The proof of the `only if' part and the proof of (ii) are similar.
\end{proof}
\medskip
	
Since a law-invariant, monetary and convex risk measure automatically satisfies the Fatou property (see \cite{JST06} and \cite{GLMX18} 
for recent developments on the automatic validity of the Fatou property on general spaces), it follows from Lemma~\ref{lemma:Fatou} that a law-invariant geometrically convex return risk measure automatically has the lower-bounded Fatou property.\ 
As a consequence, we have the following mixture continuity result, in which, as usual, we denote by $\delta_x$ a probability measure supported at $x$.

\begin{proposition}
\label{prop:mix-cont}
Let $\trho \colon \mathcal{M}_{1,c}(0, +\infty) \to (0,+\infty)$ be a law-invariant geometrically convex return risk measure. 
Let $0<x<y$. 
Then the mapping
\[
\lambda \mapsto \trho (\lambda \delta_x+(1-\lambda)\delta_y)
\]
is continuous at each $\lambda \in [0,1)$.
\end{proposition}
\begin{proof}
From Lemma~\ref{lemma:rho_trho} it follows that the corresponding $\rho \colon L^\infty \to \R$ given by equation \eqref{eq:trho-2} is a convex law-invariant monetary risk measure.\
From Proposition~2.1 in \cite{DBBZ14} suitably adapted to our sign conventions it follows that
\[
\lambda \mapsto \rho(\lambda \delta_u+(1-\lambda)\delta_v)
\]
is continuous at each $\lambda \in [0,1)$, for fixed $u,v\in \R$ with $u<v$.  
Therefore, from the representation \eqref{eq:trho_dist}, we obtain
\begin{align*}
\Tilde{\rho}(\lambda \delta_x+(1-\lambda)\delta_y)=\exp(\rho(\lambda \delta_{\log(x)}+(1-\lambda)\delta_{\log(y)})),
\end{align*}
and the thesis follows by the continuity of compositions with $\exp$ and $\log$.
\end{proof}

\section{Dual representations}\label{sec:dual}
	
In this section, we will denote by $\mathbf{P}$ the set of probability measures on $(\Omega, \F, P)$ that are absolutely continuous with respect to the reference measure $P$.\ 
In the next theorem, we derive a dual representation of geometrically convex return risk measure as suprema of suitably weighted, or discounted, logarithmic certainty equivalents; the less plausible the probabilistic model, the lower is the corresponding discount factor.

\begin{theorem}
Let $\trho \colon L^\infty_{++} \to (0,+\infty)$ be a geometrically convex return risk measure satisfying the lower-bounded Fatou property.\ 
Then there exists a multiplicative weighting function $\tilde{\alpha} \colon \mathbf{P} \to [0,1]$ satisfying $\sup_{Q \in \mathbf{P}} \tilde{\alpha}(Q)=1$
such that
\begin{equation}\label{eq:dual_rep}
\tilde{\rho}(X)=\sup_{Q\in \mathbf{P}} \{ \tilde{\alpha}(Q)\exp\of{\E_Q[\log X]} \}.
\end{equation}
Furthermore, if $\trho$ satisfies the lower-bounded Lebesgue property, the supremum in \eqref{eq:dual_rep} is attained.

\end{theorem}
\begin{proof}
From Lemma~\ref{lemma:rho_trho} and Lemma~\ref{lemma:Fatou}, it follows that $\rho(X)=\log (\trho(\exp(X)))$ is a
convex monetary risk measure with the Fatou property, so as is well-known (see e.g.,\ \cite{FS11}) it has the dual representation
\begin{equation} \label{eq:dual_rep_rho}
\rho(X)=  \sup_{Q\in \mathbf{P}} \{ \E_Q[X]-\alpha(Q) \},
\end{equation}
where $\alpha \colon \mathbf{P} \to [0,+\infty]$.
Since $\trho(X)=\exp (\rho(\log(X)))$, it follows that
\begin{align*}
\tilde{\rho}(X) &= \exp \left ( \sup_{Q\in \mathbf{P}} \{  \E_Q[\log X]-\alpha(Q)  \} \right)
=  \sup_{Q\in \mathbf{P}} \left \{  \exp \left (   \E_Q[\log X]-\alpha(Q)   \right) \right \} \\
  &=\sup_{Q\in \mathbf{P}} \left \{  \tilde{\alpha}(Q) \exp (\E_Q[\log X] )   \right \},
\end{align*}
where $\tilde{\alpha} \colon \mathbf{P} \to [0,1]$ is given by  $\tilde{\alpha}(Q)=\exp(-\alpha(Q))$.\ 
From $\trho(1)=1$, it follows that $\sup_{Q \in \mathbf{P}} \tilde{\alpha}(Q)=1$.\
By Theorem~4.22 and Exercise~4.2.2 in \cite{FS11}, the supremum in \eqref{eq:dual_rep_rho} is attained if $\rho$ has the Lebesgue property.\
In view of Lemma~\ref{lemma:rho_trho} and Lemma~\ref{lemma:Fatou}, it then follows that the supremum in \eqref{eq:dual_rep} is attained provided that $\trho$ satisfies the lower-bounded Lebesgue property. 
\end{proof}
\smallskip
	
\begin{remark} \label{rem:dual_Orl}
The logarithmic certainty equivalent $\exp \left(\E[\log X] \right)$ arising in the dual representation~\eqref{eq:dual_rep} can already be viewed as an example of an Orlicz premium corresponding to the unbounded Orlicz 
function $\Phi(x)=1+\log(x)$, since
\begin{align*}
\exp \left ( \E_Q \left [\log X \right ] \right )
= \inf \left \{ k >0 \; \Big | \; \E_Q \left [ 1+ \log \left ( \frac{X}{k} \right)  \right ] \leq 1 \right \}
= H_{1+\log x, Q} (X).
\end{align*}
We refer to Definition~\ref{def:orl} in Section~\ref{sec:Orl} for details on terminology and notation.
As a consequence, every geometrically convex return risk measure satisfying the lower-bounded Fatou property can be seen as the supremum of a suitable family of multiplicatively weighted Orlicz premia.
\end{remark}

We now derive a Kusuoka representation for law-invariant geometrically convex return risk measures that parallels the usual one for law-invariant convex risk measures.\ 
Recall first the definition of Average Value-at-Risk.
	
\begin{definition}
Let $X \in L^1(\Omega, \F, P)$. 
For $\lambda \in [0,1)$, the Average Value-at-Risk of $X$ at level $\lambda$ is given by
\[
AV@R_\lambda (X) = \frac{1}{1-\lambda} \int_\lambda^1 q_\alpha(X) \, \diff\alpha,
\]
where
\[
q_\alpha(X) = \inf \{ x \in \R \mid F(x) \geq \alpha \},
\]
and  for $\lambda=1$ we set by definition $AV@R_1 (X) = \esssup(X)$. 
\end{definition}

Denote by $\M_1([0,1])$ the set of probability measures with support in $[0,1]$.

\begin{theorem}
Let $\trho \colon L^\infty_{++} \to (0,+\infty)$ be a law-invariant geometrically convex return risk measure. 
Then there exists
$\tilde{\beta} \colon \mathcal{M}_{1}([0,1])   \to [0,1]$
such that
\begin{equation} \label{eq:kusu_rep}
\trho(X) = \sup_{\mu\in\mathcal{M}_{1}([0,1])} \left \{ \tilde{\beta}(\mu)\exp\left(\int_{[0,1]}AV@R_{\lambda}(\log X)\mu(\diff\lambda)\right) \right \}.
\end{equation}
If $\trho$ has the lower-bounded Lebesgue property, then $\mu(1)>0 \Rightarrow \tilde{\beta}(\mu)=0$. 
\end{theorem}
	
\begin{proof}
As in the proof of the previous theorem, if $\trho$ is also law invariant, then from Lemma~\ref{lemma:rho_trho}
the associated convex risk measure $\rho$ given by \eqref{eq:trho-2} is also law invariant, and hence has the Kusuoka representation (see e.g.,\ \cite{FS11}, \cite{D12})
\[
\rho(X)=\sup_{\mu\in\mathcal{M}_{1}([0,1])}\left(\int_{[0,1]}AV@R_{\lambda}(X)\mu(\diff\lambda)-\beta(\mu)\right),
\]
for a suitable $\beta \colon \mathcal{M}_{1}([0,1])   \to [0,+\infty]$, from which it follows that
\begin{align*}
\tilde{\rho}(X)&=\exp\left(\sup_{\mu\in\mathcal{M}_{1}([0,1])}\left(\int_{[0,1]}AV@R_{\lambda}(\log X)\mu(\diff\lambda)-\beta (\mu)\right)\right)\nonumber\\
&=\sup_{\mu\in\mathcal{M}_{1}([0,1])}\tilde{\beta}(\mu)\exp\left(\int_{[0,1]}AV@R_{\lambda}(\log X)\mu(\diff\lambda)\right),
\end{align*}
with $\tilde{\beta}(\mu)=\exp(-\beta (\mu))$. 
From Lemma~\ref{lemma:Fatou}, if $\trho$ has the lower-bounded Lebesgue property, then $\rho$ has the Lebesgue property, and Theorem~35 in \cite{D12} implies that $\mu(1)>0 \Rightarrow \beta(\mu)=+\infty$, from which the thesis follows. 
\end{proof}

\begin{remark}
In the same spirit of Remark~\ref{rem:dual_Orl}, the Kusuoka representation of a law-invariant geometrically convex return risk measure given in equation~\eqref{eq:kusu_rep} can be written in terms of Orlicz premia as follows:
\begin{equation*}
\trho(X)=\sup_{Q\in\mathcal{M}_{1}(P)}\{\tilde{\alpha}(Q)\sup_{\tilde{Q}\sim Q}H_{1+\log x,\tilde{Q}}(X)\},
\end{equation*}
where $\tilde{Q}\sim Q$ indicates that $\frac{\mathrm{d}Q}{\mathrm{d}P}$ and $\frac{\mathrm{d}\tilde{Q}}{\mathrm{d}P}$ have the same distribution.
\end{remark}

As for convex risk measures, the validity of the lower-bounded Lebesgue property is linked to a suitable weak compactness property of the level sets of the weighting function $\tilde{\alpha}$ and in the law-invariant case to the so-called $\psi$-weak continuity. 
Before stating the main result we recall two basic definitions.

\begin{definition}
We say that a monetary convex risk measure $\rho$ with dual representation \eqref{eq:dual_rep_rho} has the WC property if for each $m \in \R$ the lower level sets $\{Q\in \mathbf{P}\,|\, \alpha (Q)\leq m\}$  are compact in the  $\sigma(L_1,L_{\infty})$ topology.
Similarly, a geometrically convex return risk measure $\trho$ with dual representation \eqref{eq:dual_rep} has the \text{$\widetilde{WC}$} property if for each $m >0$ the upper level sets $\{Q\in \mathbf{P}\,|\,\tilde{\alpha}(Q)\geq m\}$ are compact in the  $\sigma(L_1,L_{\infty})$ topology.
\end{definition}

\begin{definition}
Let $\psi \colon \mathbb{R}\to [1,+\infty)$ be continuous.\
The $\psi$-weak topology on $\M_{1,c}(\R)$ is the weakest topology that makes all mappings $F \mapsto \int f \, \diff F$ continuous, for each continuous $f$ satisfying $\vert f \vert \leq c \psi$, with $c >0$.
It holds that
\[
F_n \overset{\psi}{\to} F \text{ if } F_n \overset{\text{weakly}}{\to} F \text{ and } \int \psi\, \diff F_n \to \int \psi\, \diff F.
\]
A functional $\rho \colon \M_{1,c}\to \R$ is $\psi$-weakly continuous if
\[F_n \overset{\psi}{\to} F \implies \rho(F_n)\to \rho(F).\]
\end{definition}

\begin{proposition}
Let $\trho: \mathcal{M}_{1,c}\to (0,+\infty)$ be a law-invariant geometrically convex return risk measure. 
The following are equivalent:
\begin{enumerate}[label={\alph*}), left = 0pt]
\item $\trho$ has the \text{$\widetilde{WC}$} property
\item $\trho$ has the lower-bounded Lebesgue property
\item  $\trho$ is $\tilde{\psi}$-weakly continuous for some $\tilde{\psi} \colon (0,+\infty) \to \R$
\item For each $x,y>0$ with $x<y$ and $\lambda \in [0,1]$, the function
\[
\lambda \mapsto \tilde{\rho}(\lambda \delta_x+(1-\lambda)\delta_y)
\]
 is continuous.
\end{enumerate}
\end{proposition}
	
\begin{proof}
If $\rho$ and $\trho$ are related by the correspondence given in \eqref{eq:trho-1} and \eqref{eq:trho-2}, then the WC property of $\rho$ is equivalent to the \text{$\widetilde{\mathrm{WC}}$} property of $\trho$, since
\begin{align*}
\{Q\in \mathbf{P}|\tilde{\alpha}(Q)\geq m\}=\{Q\in \mathbf{P}|\exp(-\alpha(Q))\geq m\}=
\{Q\in \mathbf{P}|\alpha(Q)\leq -\log(m)\}.
\end{align*}
So (a) holds if and only if the associated convex risk measure $\rho$ has the WC property.\ 
As is well-known (see e.g.,\ \cite{D12}), for convex risk measures the WC property is equivalent to the Lebesgue property, so from Lemma~\ref{lemma:Fatou} it follows that (a) is equivalent to (b).\ 
From Proposition~2.7 in \cite{DBBZ14} adapted to our sign conventions, it follows that the WC property of $\rho$ is equivalent to $\psi$-weak continuity with respect to some gauge function $\psi$.\ 
From Lemma~4 of \cite{BLR18}, it holds that $\rho$ is $\psi$-weakly continuous if and only if $\trho$ is $\tilde{\psi}$-weakly continuous with $\tilde{\psi}(t)=\psi(\log(t))$, which shows the equivalence between (a) and (c).\ 
Further, Proposition~2.7 in \cite{DBBZ14} shows that the WC property of $\rho$ is equivalent to its mixture continuity for $\lambda \to 1^-$, which combined with Proposition~\ref{prop:mix-cont} shows that (a) is equivalent to (d).
\end{proof}
\medskip

As we anticipated after Definition~\ref{def:georetrm}, geometrically convex return risk measures are a generalization of convex risk measures, as the following shows.

\begin{lemma}
\label{lemma:convex-GGconvex}
Let $\trho \colon L^\infty_{+} \to [0, +\infty)$ be a convex return risk measure.\ 
Then $\trho$ is geometrically convex.
\end{lemma}
\begin{proof}
If $\trho(X)=0$ or $\trho(Y)=0$ the thesis is trivial. 
Let $X,Y\in L_{+}^{\infty}$ and $\lambda \in (0,1).$ 
By using the AM-GM inequality and the monotonicity and convexity of $\trho$, it follows that
\begin{align*}
\trho\of{\of{\frac{X}{\trho(X)}}^{\lambda}\of{\frac{Y}{\trho(Y)}}^{1-\lambda}}&\leq \trho\of{\lambda \frac{X}{\trho(X)}+(1-\lambda)\frac{Y}{\trho(Y)}}\\&\leq \lambda \trho\of{\frac{X}{\trho(X)}}+(1-\lambda)\trho \of{\frac{Y}{\trho(Y)}}=1.
\end{align*}
Next, from positive homogeneity, it follows that
\[\trho\of{X^{\lambda}Y^{1-\lambda}}\leq \trho(X)^{\lambda}\trho(Y)^{1-\lambda},\]
which completes the proof.
\end{proof}
\medskip

It is then interesting to compare the dual representation of geometrically convex return risk measures given in equation~\eqref{eq:dual_rep} with the following dual representation for convex return risk measures.
\begin{proposition}\label{prop:crrm}
Let $\trho \colon L^\infty_{+} \to [0,+\infty)$ be a convex return risk measure satisfying the Fatou property.\
Then
there exists $\beta \colon \mathbf{P} \to [0,1]$ with $\sup_{Q\in \mathbf{P}}\beta(Q)=1$  such that
\begin{equation}\label{eq:dual_rep_conv}
\trho(X)=\sup_{Q\in \mathbf{P}} \{ \beta(Q) \E_Q[X] \}.
\end{equation}
Furthermore, if $\trho$ satisfies the Lebesgue property, the supremum in~\eqref{eq:dual_rep_conv} is attained.
\end{proposition}
\begin{proof}
The first part of the statement is easily derived from Proposition~4.3 in \cite{LS13}.
For the second part, it follows from the proof of Proposition~4.3 in \cite{LS13} that
\begin{equation*}
\trho(X)=\sup_{Z\in H}\E[XZ],    
\end{equation*}
where $H=\cb{Z\in L_+^1: \E[ZY]\leq \trho(Y) \text{ for any } Y\in L_+^{\infty}}$. 
If we take $Y=1$, then $\E[Z]\leq 1$ for any $Z\in H$, which gives the norm-boundedness of the set $H$. 
Furthermore, $H$ is weakly closed, since it is an intersection of weakly closed sets. 
Let us take a decreasing sequence $\of{A_n}_n \in \mathcal{F}$ of which the intersection is the empty set. 
For any $Z\in H$, we have $\E[Z1_{A_n}]\leq \trho(1_{A_n})$ for every $n$. 
Therefore, by using the Lebesgue property of $\trho$, we have
\begin{equation*}
\lim_{n \to +\infty}\sup_{Z\in H}\E[Z1_{A_n}]\leq \lim_{n \to +\infty} \trho \of{1_{A_n}}=0,
\end{equation*}
which gives that $H$ is uniformly integrable. 
Because $H$ is bounded, weakly closed and uniformly integrable, it is weakly compact as a consequence of the Dunford-Pettis theorem (see, e.g., Theorem~A.67 in \cite{FS11}). 
Therefore, the supremum is attained as a result of the Weierstrass Theorem (see, e.g., Corollary~2.35 in \cite{AB06}). 
Suppose the supremum is attained for $\Tilde{Z}\in H$. 
Then, the supremum is attained for $\Tilde{Q}$ such that $\frac{\diff \Tilde{Q}}{\diff P}=\frac{\Tilde{Z}}{\E[\Tilde{Z}]}$.
\end{proof}
\medskip

Since from Lemma~\ref{lemma:convex-GGconvex} a convex return risk measure is geometrically convex, it follows that the dual representation~\eqref{eq:dual_rep} is implied by the dual representation~\eqref{eq:dual_rep_conv}.\ 
This can be seen starting from the well-known dual representation of the exponential certainty equivalent (see e.g.,\ \cite{DK13}), given by
\begin{equation}\label{eq:dual_exp}
\log \E_Q[\exp(Y)] = \sup_{R \ll Q} \left\{ \E_R[Y] -H(R,Q) \right\},
\end{equation}
where $H(R,Q)$ is the relative entropy or Kullback-Leibler divergence defined by (\cite{C75})
\[
H(R,Q):=
\begin{cases}
\E_Q \left [\frac{\mathrm{d}R}{\mathrm{d}Q} \log \frac{\mathrm{d}R}{\mathrm{d}Q} \right] &\text { if } R \ll Q \\
+\infty &\text{otherwise}.
\end{cases}
\]
Letting $X=\exp (Y)$ and exponentiating both sides of \eqref{eq:dual_exp}, we obtain
\begin{equation}
\E_Q[X] =\sup_{R\ll Q} \left\{\tilde{\alpha}(R)\exp\left(\E_R[\log X] \right)\right\},\label{eq:dual_exp_2}
\end{equation}
where
\[
\tilde{\alpha}(R)=\exp \left(-H(R,Q) \right).
\]

Now note that $\tilde{\alpha}(R)=0$ when $R$ is not absolutely continuous with respect to $Q$.\ 
Using this fact, we can rewrite expression~\eqref{eq:dual_exp_2} for $Q\in \mathbf{P}$, as follows:
\begin{equation}
\E_Q[X] =\sup_{R\ll Q} \left\{\tilde{\alpha}(R)\exp\left(\E_R[\log(X)] \right)\right\}
=\sup_{R\in\mathbf{P}}\left\{\tilde{\alpha}(R)\exp\left(\E_R[\log(X)] \right)\right\},
\label{eq:dual_exp_3}
\end{equation}
since we take a supremum of nonnegative numbers, 
$\tilde{\alpha}(R)=0$ when $R\notin \mathbf{Q}$ and $\mathbf{Q}\subset \mathbf{P}$, where $\mathbf{Q}$ denotes the set of probability measures absolutely continuous with respect to $Q$. 
Upon substituting the expression for $\E_Q[X]$ derived in \eqref{eq:dual_exp_3} in \eqref{eq:dual_rep_conv}, 
we have clarified the connection between \eqref{eq:dual_rep} and \eqref{eq:dual_rep_conv}, as follows:
\begin{align*}
\trho(X) &=\sup_{Q\in \mathbf{P}} \{ \beta(Q) \E_Q[X] \}\\
&=\sup_{Q\in \mathbf{P}} \left\{ \beta(Q) \sup_{R\ll Q} \left\{ \tilde{\alpha}(R)\exp\left(\E_R[\log X]\right)\right\} \right\} \\
&=\sup_{Q\in \mathbf{P}} \left\{ \beta(Q) \sup_{R\in \mathbf{P}} \left\{ \tilde{\alpha}(R)\exp\left(\E_R[\log X]\right)\right\} \right\} \\
&=\sup_{Q\in \mathbf{P}} \sup_{R\in \mathbf{P}} \{ \beta(Q) \left\{ \exp\of{-H(R,Q)}\exp\left(\E_R[\log X]\right)\right\} \} \\
&=\sup_{R\in \mathbf{P}} \sup_{Q\in \mathbf{P}} \{ \beta(Q) \left\{ \exp\of{-H(R,Q)}\exp\left(\E_R[\log X]\right)\right\} \} \\
&=\sup_{R\in \mathbf{P}} \left\{ c(R)\exp\left(\E_R[\log X]\right)\right\},
\end{align*}
where
\[
c(R)=\sup_{Q\in \mathbf{P}}\beta(Q)\exp\of{-H(R,Q)}.
\]

Following the construction outlined in \cite{BT86,BT07} and \cite{BRG08}, 
return risk measures may be optimized and become translation invariant, hence monetary, as follows: 
\begin{definition}
\label{def:OR}
An \textit{optimized return risk measure} (henceforth, OR risk measure) $\rho: L_{+}^{\infty}\to \R$ is defined as
\begin{equation}\label{eq:OR}
\rho(X)=\inf_{x \in \R}\cb{x+\trho\of{\of{X-x}^+}},
\end{equation}
for a corresponding return risk measure $\trho:L^{\infty}_{+}\rightarrow[0,+\infty)$.
\end{definition}

\begin{lemma}
An OR risk measure satisfies the following properties:
\begin{itemize}
\item [a)] monotonicity
\item [b)] positive homogeneity
\item [c)] translation invariance
\item [d)] if $\trho$ is convex, then $\rho$ is convex.
\end{itemize}
\end{lemma}
\begin{proof}
Take $X,Y\in L_{+}^{\infty}$ such that $X\leq Y$.
For an arbitrary $x\in \R$, we have $(Y-x)^+\geq (X-x)^+$, which implies $x+\trho((Y-x)^+)\geq x+\trho((X-x)^+)$ due to the monotonicity of $\trho$. 
Since this is valid for any $x\in \R$, by taking the infimum on both sides, we obtain $\rho(X)\leq \rho(Y)$.
For~(b), by using the positive homogeneity of $\trho$ and of the positive part function, we have, for any $\lambda >0$,
\begin{align*}
\rho(\lambda X)&=\inf_{x \in \R}\cb{x+\trho\of{\of{\lambda X-x}^+}}=\inf_{x \in \R}\cb{x+\lambda \trho\of{\of{X-\frac{x}{\lambda}}^+}}\\
&=\inf_{\Tilde{x}\in \R}\cb{\lambda \Tilde{x}+\lambda \trho\of{\of{X-\Tilde{x}}^+}}
=\lambda \inf_{\Tilde{x}\in \R}\cb{ \Tilde{x}+ \trho\of{\of{X-\Tilde{x}}^+}}=\lambda \rho(X).
\end{align*}
For~(c), we have, for any $h\in \R$,
\begin{align*}
\rho(X+h)&=\inf_{x \in \R}\cb{x+\trho\of{\of{X+h-x}^+}}=\inf_{x \in \R}\cb{x+\trho\of{\of{X-(x-h)}^+}}\\
&=\inf_{\Tilde{x} \in \R}\cb{\Tilde{x}+h+\trho\of{\of{X-\Tilde{x}}^+}}
=h+\inf_{\Tilde{x} \in \R}\cb{\Tilde{x}+\trho\of{\of{X-\Tilde{x}}^+}}\\
&=\rho(X)+h.
\end{align*}
Finally, let us assume that $\trho$ is convex
and take $X,Y \in L_{+}^{\infty}$.
Because $\rho$ is positively homogeneous, it is sufficient for~(d) to prove that $\rho$ is subadditive.
We have
\begin{align*}
\rho(X+Y)&=\inf_{x,y\in \R}\cb{x+y+\trho\of{\of{X+Y-x-y}^+}}\\
&\leq\inf_{x,y\in \R}\cb{x+y+\trho \of{\of{X-x}^+}+\trho \of{\of{Y-y}^+}}\\
&=\inf_{x \in \R}\cb{x+\trho\of{\of{X-x}^+}}+\inf_{y \in \R}\cb{y+\trho\of{\of{Y-y}^+}}\\
&=\rho(X)+\rho(Y),
\end{align*}
where we have used the convexity and positive homogeneity of $\trho$ and of the positive part function in the second line.
\end{proof}

\medskip 
The class of OR risk measures encompasses as special cases the
Rockafellar-Uryasev \cite{RU00} construction of Average-Value-at-Risk
as well as its generalization given by the (robust) HG risk measure (\cite{BLR18}).
We now establish that the OR risk measure admits an inf-convolution and a dual representation.

\begin{definition}
The inf-convolution $(f\ \Box\ g)$ of two convex functionals $f:L^{\infty}\to \overline{\R}$ and $g:L^\infty \to \overline{\R}$ is defined as follows:
\begin{equation*}
(f\ \Box\ g)(X)=\inf_{Y\in L^{\infty}}\cb{f(X-Y)+g(Y)}.
\end{equation*}
\end{definition}
\begin{lemma}
\label{lem:ORcon}
An OR risk measure $\rho$ can be written as
\begin{equation*}
\rho(X)=(f\ \Box\ g)(X),
\end{equation*}
where $f(X)=\trho(X^+)$ and
\begin{equation*}
g(Y)=\begin{cases}
        x & \text{ if } Y=x,\\
        +\infty & \text{otherwise},
    \end{cases}
\end{equation*}
when the corresponding return risk measure $\trho$ is convex.
\end{lemma}
\begin{proof}
Note that the functional $f$ is convex since $\trho$ is convex and monotone and the positive part function is convex, and $g$ is convex, too.
Then, the inf-convolution of the functionals $f$ and $g$ agrees with the definition of $\rho$ in~\eqref{eq:OR}.
\end{proof}
\medskip

Recall that the dual space of $L^\infty$ can be identified with $L^1$ w.r.t.~the $\sigma(L^\infty,L^1)$-topology. 
The convex conjugate $h^*:L^1\to \overline{\R}$ of a function $h:L^{\infty}\to \overline{\R}$ is defined as:
\begin{equation*}
h^*(\varphi)=\sup_{X\in L^\infty}\cb{\E[\varphi X]-h(X)}.
\end{equation*}

\begin{proposition}
\label{prop:ORdualrep}
An OR risk measure $\rho$, with a corresponding convex return risk measure $\trho$, admits the following dual representation:
\begin{equation}
\label{eq:ORdualrep}
\rho(X)=\sup_{Q\in A_P}\E_Q[X],
\end{equation}
where $A_P=\cb{Q\in \mathbf{P}:\E_Q[Y]\leq \trho(Y) \text{ for any } Y\in L_+^\infty}$.
Furthermore, if $\trho$ satisfies the Lebesgue property, then the supremum in \eqref{eq:ORdualrep} is attained.
\end{proposition}
\begin{proof}
From, e.g., Theorem~2.3.1 in \cite{Z02}, 
it is known that
\begin{equation}
\label{eq:conjugatecon}
(f\ \Box\ g)^*=f^*+g^*.
\end{equation}
Let us consider the conjugates of the functionals $f$ and $g$ in Lemma~\ref{lem:ORcon}.
The conjugate $f^*$ can be calculated as follows:
\begin{align*}
f^*(\varphi)&=\sup_{Y\in L^{\infty}}\cb{\E[\varphi Y]-f(Y)}=\sup_{Y\in L^{\infty}}\cb{\E[\varphi Y]-\trho(Y^+)}\\
&=\begin{cases}
        0 & \text{if } \varphi\in L^{1}_{+} \text{ and } \E[\varphi Y] \leq \trho (Y) \text{ for any } Y\in L_+^\infty,\\
        +\infty & \text{otherwise},
    \end{cases}
\end{align*}
by using the positive homogeneity of $\trho$.
The conjugate $g^*$ can be calculated as follows:
\begin{align*}
g^*(\varphi)&=\sup_{Y\in L^{\infty}}\cb{\E[\varphi Y]-g(Y)}=\sup_{x\in \R}\cb{\E[\varphi x]-x}=\sup_{x\in \R}x(\E[\varphi]-1)\\
&=\begin{cases}
        0 & \text{if } \E[\varphi]=1,\\
        +\infty & \text{otherwise}.
    \end{cases}
\end{align*}
By using~\eqref{eq:conjugatecon} and Lemma~\ref{lem:ORcon}, we have the following:
\begin{align*}
\rho^*(\varphi)
=\begin{cases}
        0 & \text{if } \varphi\in L^{1}_{+},\ \E[\varphi]=1 \text{ and } \E[\varphi Y] \leq \trho (Y) \text{ for any } Y\in L_+^\infty,\\
        +\infty & \text{otherwise}.
    \end{cases}
\end{align*}
Therefore, $\rho^*$ is the indicator function of the set \[A=\cb{\varphi \in L^{1}_{+} : \E[\varphi]=1 \text{ and } \E[\varphi Y] \leq \trho (Y) \text{ for any } Y\in L_+^\infty}.\]
Since the functional $\rho$ is the inf-convolution of the functionals $f$ and $g$, as a consequence of the Fenchel-Moreau theorem, we have
\begin{align*}
\rho(X)&=\sup_{\varphi\in L_1}\cb{\E[\varphi X]-\rho^*(\varphi)}
=\sup_{\varphi \in A}\E[\varphi X] \nonumber\\
&=\sup_{Q\in A_P}\E_Q[X],
\end{align*}
where $A_P=\cb{Q\in \mathbf{P}:\E_Q[Y]\leq \trho(Y) \text{ for any } Y\in L_+^\infty}$.
Hence,
\begin{equation}
\label{eq:ORdualrep-b}
\rho(X)=\sup_{Q\in A_P}\E_Q[X].
\end{equation}

Following the same argument used in the proof of Proposition~\ref{prop:crrm}, it follows that if $\trho$ has the Lebesgue property, then the set $A_P$ is weakly compact, from which the attainment of the maximum follows.
Indeed, the set $A_P$ is norm-bounded by definition. 
It is weakly closed, since it is the intersection of weakly closed sets. Now let us take a decreasing sequence $(A_n)_n\in \mathcal{F}$ of which the intersection is the empty set. 
For any $Q \in A_P$, we have $\E_Q[ 1_{A_n}]\leq \trho(1_{A_n})$ for every $n$. 
Hence, by using the Lebesgue property of $\trho$, we have
\[\lim_{n\to +\infty}\sup_{Q \in A_P}\E_Q[ 1_{A_n}]\leq \lim_{n\to +\infty}\trho(1_{A_n})=0,\]
which gives that $A_P$ is uniformly integrable. 
Therefore, the supremum is attained due to the Dunford-Pettis and Weierstrass theorems; cf.~also Theorem~3.6 in \cite{D02} and Theorem~8 in \cite{D00}. 
\end{proof}

\section{Axiomatizations of Orlicz premia}\label{sec:Orl}
	
In this section, we first define Orlicz premia and derive some properties that are relevant in this paper; and next we establish new axiomatizations of Orlicz premia and compare them with the one given in \cite{BLR18}.

\subsection{Orlicz premia: Definition and properties}

The mathematical definition of the Orlicz premium corresponds to the Luxemburg norm on the Orlicz space
\[
L^{\Phi}:=\cb{X\in L^0(\Omega,\F,P): \E\sqb{\frac{\abs{X}}{k}}<+\infty \text{ for some } k>0},
\]
given by
\[
H_\Phi(X) := \inf \{k > 0 \mid  \E \left [\Phi(X/k ) \right] \leq 1 \},
\]
where the Orlicz function $\Phi \colon [0, +\infty) \to [0, +\infty]$ satisfies $\Phi(0)=0$, is nondecreasing, left-continuous, convex, and nontrivial in the sense that $\Phi(x) >0$ for some $x>0$ and $\Phi(x) < +\infty$ for some $x>0$.\
We refer to \cite{ES92} for the basic properties of Luxemburg norms under these assumptions.\ 
Notice that in the actuarial and financial mathematics literature (e.g., \cite{HG82}, \cite{CL08}, \cite{CL09}, \cite{BLR18}, \cite{BLR21}) there are small differences in the set of properties required to $\Phi$.\
In this paper, we are interested in possibly nonconvex Orlicz functions that may not satisfy $\Phi(0)=0$; conversely, we will limit the domain of $H_\Phi$ to $L^{\infty}_{+}$. 
(When the function $\Phi(\cdot)$ is convex and satisfies several additional properties, it is often referred to as a Young function; as these conditions are not assumed in this paper, we refer to $\Phi(\cdot)$ as an Orlicz function.)
This leads to the following definition.

\begin{definition}\label{def:orl}
Let $\Phi \colon [0, +\infty) \to \R$ satisfy:
\begin{enumerate}[left=0pt, itemsep=0pt]
\item [a)] $\Phi(1)=1$, $\lim_{x \to +\infty}\Phi(x)=+\infty$
\item [b)] $\Phi$ is nondecreasing
\item [c)] $\Phi$ is left-continuous
\end{enumerate}
For $X \in L^{\infty}_{+}$, the Orlicz premium is defined by
\[
H_\Phi(X) = \inf \{k > 0 \mid  \E \left [\Phi(X/k ) \right] \leq 1 \}.
\]
\end{definition}
We recall the relevant properties of Orlicz premia in the following proposition.

\begin{proposition}\label{prop:orl}
Let $\Phi \colon [0, +\infty) \to \R$ and $H_\Phi(X)$ be as in Definition~\ref{def:orl}. 
Then,
\begin{enumerate}[left=0pt, itemsep=0pt]
\item [a)] $H_\Phi$ is monotone, positively homogeneous and satisfies $H_\Phi (1) = 1$
\item [b)] for each $X \in L^{\infty}_{++}$, it holds that $\E \left [\Phi(X/H_\Phi(X) ) \right]=1$
\item [c)] if $\Phi$ is increasing, then
\begin{align*}
&\E \left [\Phi(X/k ) \right] = 1 \iff k=H_\Phi(X) \\
&\E \left [\Phi(X/k ) \right] > 1 \iff k<H_\Phi(X)
\end{align*}
\item [d)] $H_\Phi$ is convex if and only if $\Phi$ is convex.
\end{enumerate}
\end{proposition}
\begin{proof}
(a) The proof is standard. 
(b) Let $g(k):=\E \left[\Phi(X/k) \right]$. 
If $k_n \downarrow k$ then $\Phi(X/k_n) \uparrow \Phi(X/k)$, so from
the monotone convergence theorem it follows that $g$ is right-continuous.
Since
$
H_\Phi(X) = \inf \{ k \mid g(k) \leq 1 \}
$,
it follows that $g \left(H_\Phi(X)\right)=1$, that is, $\E \left [\Phi(X/H_\Phi(X) ) \right] = 1$.
(c) The `only if' part of the first implication follows by strict monotonicity of $\Phi$. 
The `if' part of the second implication is just the definition of $H_\Phi$, while the `only if' part follows from~(c). 
(d) The proof of the `if' part is standard. 
To prove the `only if' part, assume first by contradiction that $\Phi$ is not midconvex, i.e., there exist $x_1,x_2 \geq 0$ such that $\Phi \left ((x_1+x_2)/2 \right )>(\Phi(x_1)+\Phi(x_2))/2$. 
Then, there exists $z \in [0,+\infty) $ and $\lambda \in (0,1)$ such that
\begin{equation}\label{eq:convexity}
\lambda \Phi(z)+(1-\lambda)\Phi\left ((x_1+x_2)/2 \right )>1>\lambda \Phi(z)+(1-\lambda) \frac{\Phi(x_1)+\Phi(x_2)}{2}.
\end{equation}
Let $A,B,C \in \F$ be disjoint sets with $P(A)=\lambda$, $P(B)=P(C)=\frac{1-\lambda}{2}$ and let
\begin{align*}
X&=z 1_A+x_1 1_B+ x_2 1_C\\
Y&=z 1_A+x_2 1_B+ x_1 1_C\\
Z&=z 1_A+\frac{x_1+x_2}{2} 1_{B\cup C}=\frac{X+Y}{2}.
\end{align*}
From \eqref{eq:convexity}, we have $\E[\Phi(Z)]>1$ and $\E[\Phi(X)]=\E[\Phi(Y)]<1$, which contradicts the convexity of $H_{\Phi}$. 
As a consequence, $\Phi$ is midconvex and since it is nondecreasing it is also convex.
\end{proof}
\medskip

A remarkable example in which $\Phi(0) \neq 0$ and $\Phi$ is not differentiable is the following. 
\begin{example}[Expectiles]\label{ex:expectiles}
Let $0 < q <1$ and let
\begin{equation*}
\Phi_q(x)=1+ q(x-1)^{+} -(1-q)(x-1)^{-}.
\end{equation*}
Then, $\Phi_q(0)=q$, $\Phi_q(1)=1$ and $\Phi_q$ is convex if $1/2 \leq q <1$ and concave if $0 < q \leq 1/2$. 
The corresponding Orlicz premium $H_{\Phi_q}$ satisfies
\[
\E[\Phi_q(X/H_{\Phi_q}(X))-1]=0,
\]
which gives
\[
\E[q(X-H_{\Phi_q}(X))^{+}-(1-q)(X-H_{\Phi_q}(X))^{-}]=0,
\]
so $H_{\Phi_q}(X)$ coincides with the $q$-expectile of $X$ (\cite{NP87,KYH09}), denoted by $e_q(X)$ and defined for $X \in L^1(\Omega, \F, P)$ by the condition
\[
q \E[(X-e_q(X))^{+}]=(1-q)\E[(X-e_q(X))^{-}].
\]
\end{example}
The following theorem shows that expectiles are the most general translation invariant convex Orlicz premia.
\begin{theorem}\label{th:expectiles}
If $\Phi$ is increasing and $H_{\Phi}$ is translation invariant and convex, then
\[
\Phi(x)=1+a(x-1)^{+}-b(x-1)^{-},
\]
with $a>b$ and $b<1$.
\end{theorem}
\begin{proof}
Let $h(x):=\Phi(x+1)-1$. Then $h(x)=0 \iff x=0$ and  $h(x)>0 \iff x>0$, and
\begin{align*}
\E[h((Y/k)-1)]&=0 \iff H_{\Phi}(Y)=k\\
\E[h((Y/k)-1)]&<0 \iff H_{\Phi}(Y)<k.
\end{align*}
Fix $x<0$ and $z>0$. 
Let $p$ be a solution of the equation
\begin{equation}\label{def:p}
ph(x)+(1-p)h(z)=0,
\end{equation}
and let $Y=p\delta_{x+1}+(1-p)\delta_{z+1}$. 
It follows that $H_{\Phi}(Y)=1$, and from translation invariance it follows that $H_{\Phi}(Y+c)=c+1$ for each $c \in \R_+$, which implies
\begin{align}\label{eq:TI}
0=\E\sqb{h\of{\frac{Y+c}{c+1}-1}}=ph\of{\frac{x}{c+1}}+(1-p)h\of{\frac{z}{c+1}}.
\end{align}
Let $\lambda=\frac{1}{c+1}$ and note that $0 < \lambda\leq 1$. 
By combining \eqref{eq:TI} with \eqref{def:p}, we have
\begin{align*}
\lambda p h(x)+\lambda (1-p)h(z)=0=ph(\lambda x)+(1-p)h(\lambda z),
\end{align*}
which gives
\[
p(\lambda h(x)-h(\lambda x))+(1-p)(\lambda h(z)-h(\lambda z))=0.
\]
From the convexity of $h$ it follows that $\lambda h(x)-h(\lambda x)\geq 0$ and $\lambda h(z)-h(\lambda z)\geq 0$, so from the last equality we get $h(\lambda z)=\lambda h(z) $ and $h(\lambda x)=\lambda h(x) $ for every $0\leq \lambda \leq 1$ and for each $x<0$ and $z>0$, from which the thesis follows.
\end{proof}
\medskip 

It is immediate to check that $H_\Phi(X)=e_q(X)$, with $q=a/(a+b)$.\ 
Further, the same argument also shows that expectiles with $0<q\leq 1/2$ are the only concave translation invariant Orlicz premia.
It is interesting to compare the theorem above with \cite{HG82} and \cite{GDH84}, where it is shown that a translation invariant Orlicz premium must be equal to the mean, but in their result actually also the differentiability of the Orlicz function $\Phi$ is assumed.
	
\begin{definition}
A function $f \colon [0,+\infty) \to \R$ is called GA-convex if, for each $\lambda \in (0,1)$ and $x,y>0$, it holds that
\[
f(x^{\lambda}y^{1-\lambda})\leq \lambda f(x)+(1-\lambda)f(y).
\]
\end{definition}
	
\noindent
It is not difficult to verify that a nondecreasing and convex function on $(0,+\infty)$ is GA-convex.\ 
For completeness we report the proof in Lemma~\ref{lemma:AA-GA} in the Appendix.\ 
The converse does not hold, an example being $f(x)=\log x$ that is increasing and GA-convex but not convex.\ 
We refer to \cite{N00} for further properties of GA-convex functions.
	
\begin{proposition}\label{prop:gg-GA}
Let $\Phi \colon [0, +\infty) \to \R$ and $H_\Phi(X)$ be as in Definition~\ref{def:orl}. 
Then $H_\Phi$ is geometrically convex if and only if $\Phi$ is GA-convex.
\end{proposition}

\begin{proof}
We first prove the `if' part.\ 
Notice first that, for each $X \in L^{\infty}_{+}$ and each $k > H_\Phi(X)$, it holds by definition that $\E[\Phi(X/k)]\leq 1$. 
Since $\Phi$ is nondecreasing and left-continuous an application of the monotone convergence theorem shows that $\E[\Phi(X/H_\Phi(X))]\leq 1$. 
Let now $X,Y\in L_{+}^{\infty}$ and $\lambda \in (0,1)$.
From the GA-convexity of $\Phi$ it follows that
\begin{align*}
&\E\sqb{\Phi\of{\of{\frac{X}{H_\Phi(X)}}^{\lambda}\of{\frac{Y}{H_\Phi(Y)}}^{1-\lambda}}} \\
&\leq  \lambda \E\sqb{\Phi\of{\frac{X}{H_\Phi(X)}}}+(1-\lambda)\E\sqb{\Phi\of{\frac{Y}{H_\Phi(Y)}}}  \leq 1,
\end{align*}
which implies
\[
H_\Phi \left ( \frac{X^{\lambda}Y^{1-\lambda}}{H_\Phi(X)^{\lambda}H_\Phi(Y)^{1-\lambda}} \right) \leq 1,
\]
which from positive homogeneity gives
\[
H_\Phi \of{X^{\lambda}Y^{1-\lambda}}\leq H_\Phi (X)^{\lambda}H_\Phi (Y)^{1-\lambda}.
\]
To prove the `only if' part, we first assume by contradiction that $\Phi$ is not GA-midconvex, i.e., there exist $x_1,x_2 \geq 0$ such that $\Phi \left (\sqrt{x_1 x_2} \right )>(\Phi(x_1)+\Phi(x_2))/2.$ Then, there exist $z \in [0,+\infty) $ and $\lambda \in (0,1)$ such that
\begin{equation}\label{eq:GA-midconvexity}
\lambda \Phi(z)+(1-\lambda)\Phi(\sqrt{x_1 x_2})>1>\lambda \Phi(z)+(1-\lambda) \frac{\Phi(x_1)+\Phi(x_2)}{2}.
\end{equation}
Take disjoint sets $A,B,C \in \F$ with $P(A)=\lambda$, $P(B)=P(C)=\frac{1-\lambda}{2}$ and let
\begin{align*}
X&=z 1_A+x_1 1_B+ x_2 1_C\\
Y&=z 1_A+x_2 1_B+ x_1 1_C\\
Z&=z 1_A+\sqrt{x_{1}x_{2}}1_{B\cup C}=\sqrt{XY}.
\end{align*}
From \eqref{eq:GA-midconvexity}, we have $\E[\Phi(Z)]>1$ and $\E[\Phi(X)]=\E[\Phi(Y)]<1$, which contradicts with the geometric convexity of $H_{\Phi}$.\ 
As a consequence, $\Phi$ is GA-midconvex. 
Since $\Phi$ is also nondecreasing the thesis follows.
Indeed, this is seen as follows.
By the induction hypothesis, $\Phi$ is rationally GA-convex. 
Now let us take $x,y\geq 0$ and $\lambda \in (0,1)$. 
Without loss of generality, assume that $x>y$. 
Take a $q\in\mathbb{Q}\cap [0,1]$ such that $q>\lambda$. 
By monotonicity and rational GA-convexity of $\Phi$, we have
\begin{equation*}
\Phi(x^{\lambda}y^{1-\lambda})\leq \Phi(x^qy^{1-q})\leq q\Phi(x)+(1-q)\Phi(y).
\end{equation*}
Since this inequality is valid for any $q\in\mathbb{Q}\cap [0,1]$ such that $q>\lambda$, we can take the infimum over the set $\Lambda_Q:=\{q\in \mathbb{Q}:1\geq q>\lambda\}$ and obtain 
\begin{equation*}
\Phi(x^{\lambda}y^{1-\lambda})
\leq\inf_{q\in \Lambda_Q}q \Phi(x)+(1-q)\Phi(y)
=\lambda \Phi(x)+(1-\lambda)\Phi(y),
\end{equation*}
which gives the GA-convexity of $\Phi$.
\end{proof}
\medskip
	
Since a nondecreasing and convex function is GA-convex, it follows that a convex Orlicz premium is also geometrically convex.\ The converse does not hold, an example being the logarithmic certainty equivalent, which is also the Orlicz premium corresponding to $\Phi(x)=1+\log(x)$.
	
\subsection{Axiomatization based on the properties of the multiplicative acceptance set}
This is Theorem~2 in \cite{BLR18} that we report below for convenience.
\begin{theorem}
Let $\trho \colon \M_{1,c}(0,+\infty)  \to \R$ be a law-invariant return risk measure and let
$\B_{\tilde{\rho}}$ be the corresponding multiplicative acceptance set as in Definition~\ref{def:ret-rm} (now at the level of distributions).  
Assume that
\begin{enumerate}[left=0pt, itemsep=0pt]
\item [a)] $\B_{\tilde{\rho}}$ and  $\B^c_{\tilde{\rho}}$ are convex with respect to mixtures, i.e., for each $\lambda\in(0,1)$, $F,G\in\B_{\tilde{\rho}}$ $\Rightarrow \lambda F+(1-\lambda)G \in \B_{\tilde{\rho}}$, and similarly for $\B^c_{\tilde{\rho}}$ 
\item [b)] $\B_{\tilde{\rho}}$ is $\tpsi$-weakly closed for some gauge function $\tpsi$
\item [c)] for each $0 <\tilde{x}<1$ and $\tilde{y} >1$, there exists $\alpha \in (0,1)$ such that
\[
\alpha \delta_{\tilde{x}} + (1-\alpha) \delta_{\tilde{y}} \in \B_{\tilde{\rho}}.
\]
\end{enumerate}
Then there exists a nondecreasing function $\Phi$ that satisfies \\
$\Phi (0) < 1 < \Phi (+\infty)$ such that $\trho(F)=H_\Phi(F)$.
\end{theorem}
As we will see, the convexity with respect to mixtures (at the level of distributions) of the multiplicative acceptance set and its complement assumed in item~(a) in the theorem above, is implied by the CxLS property.

\subsection{Axiomatizations based on CxLS}
\begin{definition}
A law-invariant functional $\rho$ has the CxLS property if
\[
\rho(F)=\rho(G)=\gamma \Rightarrow \rho(\lambda F +(1-\lambda)G)=\gamma,
\]
for each $\gamma \in \R$, $F,G \in \M_{1,c} $ and $\lambda \in (0,1)$.
\end{definition}

\begin{lemma}
Let $\trho$ be a law-invariant return risk measure with CxLS. 
Then $\mathcal{B}_{\trho}$ and $\mathcal{B}_{\trho}^c$ are convex with respect to mixtures.
\end{lemma}
\begin{proof}
Let us take $F,G \in \B_{\trho}^c$ and $\lambda \in (0,1)$. 
Let $X,Y \in L_{++}^{\infty}$ such that the distributions of $X$ and $Y$ are $F$ and $G$. 
Take $A\in \mathcal{F}$ such that $P(A)=\lambda$ and $X,Y$ and $A$ are independent. 
Choosing such $X,Y$ and $A$ is possible because we are working in an atomless probability space, see Lemma~3.1 in \cite{DBBZ14}. 
Without loss of generality, assume that $\trho(X)=k\trho(Y)$ for some $k\geq 1$. 
Define $X'=X/k$ and denote its distribution by $F'$.
By positive homogeneity, we have $\trho(X')=\trho(Y)$. 
Then, $X1_A+Y1_{A^c}$ has distribution $\lambda F+(1-\lambda)G$ and $X'1_A+Y1_{A^c}$ has distribution $\lambda F'+(1-\lambda)G$. 
Since $k\geq 1$ and $X\in L_{++}^{\infty}$, we have
\[X1_A+Y1_{A^c}\geq X'1_A+Y1_{A^c}.\]
By using the monotonicity and the CxLS property, we have
\[\trho(X1_A+Y1_{A^c})\geq \trho(X'1_A+Y1_{A^c})=\trho(Y)>1,\]
which gives the convexity of $\B_{\trho}^c$ with respect to mixtures. 
Similarly, it can be proved that $\B_{\trho}$ is convex with respect to mixtures.
\end{proof}

\begin{theorem}\label{th:axiom_2}
Let $\trho \colon L_{++}^{\infty}\to (0,+\infty)$ be a law-invariant geometrically convex return risk measure with CxLS. 
Then there exists a nondecreasing GA-convex Orlicz function $\Phi \colon [0,+\infty) \to \R \cup \{+\infty\}$ such that $\trho(X)=H_\Phi(X)$.
\end{theorem}

\begin{proof}
From the hypotheses and Lemma~\ref{lemma:rho_trho}, it follows that the corresponding $\rho$ given by \eqref{eq:trho-2} is a convex law-invariant monetary risk measure with CxLS. 
From Theorem~3.10 in \cite{DBBZ14}, there exists a convex function $\varphi \colon \R \to \R \cup \{+\infty\}$ such that $\rho(X)\leq 0$ if and only if $\E[\varphi(X)]\leq 0$. 
Letting $\Phi (x):=1+\varphi(\log(x))$, it follows that
\begin{align*}
\trho(X/k)\leq 1 &\iff \rho(\log(X/k))\leq 0 \iff \E[\varphi(\log(X/k))]\leq 0 \\
&\iff \E[\Phi (X/k)]\leq 1.
\end{align*}
From the convexity of $\varphi$, it follows that for each $x,y>0$ and $\lambda \in (0,1)$,
\begin{align*}
\Phi(x^{\lambda}y^{1-\lambda})&=1+\varphi(\log(x^{\lambda}y^{1-\lambda})=1+\varphi(\lambda \log(x)+(1-\lambda)\log(y))\\&\leq 1+\lambda \varphi(\log(x))+(1-\lambda)\varphi(\log(y))=\lambda \Phi(x)+(1-\lambda)\Phi(y),
\end{align*}
which shows the GA-convexity of $\Phi$.
\end{proof}
\medskip

Notice the consistency between Proposition~\ref{prop:gg-GA} and Theorem~\ref{th:axiom_2}.
Notice also that under the assumptions of Theorem~\ref{th:axiom_2}, the function $\Phi$ is not necessarily convex as the example of logarithmic certainty equivalents shows.

\begin{theorem}\label{th:axiom_3}
Let $\trho \colon L_{++}^{\infty}\to (0,+\infty)$ be a law-invariant convex return risk measure with CxLS. 
Then there exists a nondecreasing convex Orlicz function $\Phi \colon [0,+\infty) \to \R \cup \{+\infty\}$ such that $\trho(X)=H_\Phi(X)$.
\end{theorem}

\begin{proof}
Since a positively homogeneous, monotone convex functional defined on $L_{++}^{\infty}$ is geometrically convex, it follows from Theorem~\ref{th:axiom_2} that there exists a nondecreasing GA-convex Orlicz function $\Phi \colon [0,+\infty) \to \R \cup \{+\infty\}$ such that $\trho(X)=H_\Phi(X)$. 
Since $H_\Phi$ is convex only if $\Phi$ is convex, as has been shown in Proposition~\ref{prop:orl}, the thesis follows.
\end{proof}

\subsection{Axiomatization based on identifiability}
\begin{definition}
We say that a functional $\rho \colon  \M_{1,c} (0,+\infty ) \to (0, +\infty)$ is \emph{identifiable} if there exists at least one \emph{identification function} $I \colon (0, +\infty) \times (0, +\infty) \to \R$ that satisfies, for each $F \in  \M_{1,c}(0,+\infty)$,
\begin{align*}
\int I(x,y)\, \mathrm{d}F(y) = 0 \iff x=\rho(F),\\
\int I(x,y)\, \mathrm{d}F(y) > 0 \iff x<\rho(F).\\
\end{align*}
\end{definition}
\vspace{-0.5cm}
\begin{example}
Let $\rho(F) = \E[F]$. 
Then $\rho$ is identifiable and two possible identification functions are $I_1(x,y)=y-x$ and $I_2(x,y)=y/x-1$.
\end{example}

An identifiable functional has the CxLS property, since
\begin{align*}
&\rho(F)=\rho(G)=\gamma \Rightarrow \int I(\gamma,y)\, \mathrm{d}F(y) = 0, \int I(\gamma,y)\, \mathrm{d}G(y) = 0 \\
&\Rightarrow \int I(\gamma,y) \,\mathrm{d} (\lambda F + (1-\lambda)G) = 0 \Rightarrow  \rho (\lambda F + (1-\lambda) G) = \gamma,
\end{align*}
for each $\lambda \in (0,1)$.

\begin{theorem} \label{th:ident}
Let $\trho \colon \M_{1,c} (0,+\infty ) \to (0, +\infty)$ be an identifiable and positively homogeneous functional satisfying $\trho(1)=1$. Then there exists $\Phi \colon [0,+\infty) \to \R$ satisfying  $\Phi (1)=1$ and $\Phi (x)>1 \iff x>1$ such that
\begin{equation}
\begin{split}\label{eq:Phi}
\int \Phi(y/x)\, \mathrm{d}F(y) = 1 \iff x=\trho(F), \\
\int \Phi(y/x)\, \mathrm{d}F(y) > 1 \iff x<\trho(F). 
\end{split}
\end{equation}
Furthermore, under these assumptions, if $\trho$ is monotone then $\Phi$ is nondecreasing
and if $\trho$ is monotone and convex then $\Phi$ is nondecreasing and convex.
\end{theorem}

The proof of Theorem~\ref{th:ident} is based on the following lemma, the proof of which is postponed to the Appendix.

\begin{lemma}\label{lemma:ident}
Let $\trho \colon \M_{1,c} (0,+\infty ) \to (0, +\infty)$ be positively homogeneous with $\trho(1)=1$ and
identifiable by the function $I(x,y) \colon (0,+\infty )  \times (0,+\infty )  \to \R$. Then there exists $g \colon (0,+\infty) \to \R$ with $g(1) = 0$ and $g(t) > 0 \iff t>1$ with
\begin{align*}
\int g(y/x)\,\mathrm{d}F(y) = 0 \iff x=\trho(F),\\
\int g(y/x)\,\mathrm{d}F(y) > 0 \iff x<\trho(F).\\
\end{align*}
\end{lemma}

\begin{proof}[Proof of Theorem \ref{th:ident}]
By applying Lemma~\ref{lemma:ident} and letting $\Phi(t)=1+g(t)$, equations \eqref{eq:Phi} follow.
To prove monotonicity of $\Phi$, take $x_1, x_2 \in [0, +\infty)$ with $x_1 < x_2$. 
If  $x_1 < 1 < x_2$, then $\Phi(x_1) < 1 < \Phi(x_2)$ so monotonicity is trivial. Assume that $1 < x_1 < x_2$. Take any $z < 1$ and set
\[
\lambda:=\frac{\Phi(x_2)-1}{\Phi(x_2)-\Phi(z)} \in (0,1).
\]
Take $A \in \F$ with $P(A)=\lambda$ and let $X_1:= z1_A + x_11_{A^c} $ and $X_2:= z1_A + x_21_{A^c} $.\ 
By construction, $X_1 \leq X_2$ and $\E[\Phi(X_2)]=1$. 
From the monotonicity of $\trho$ it follows that  $\trho(X_1) \leq \trho(X_2) =1$, so $\E[\Phi(X_1)]\leq 1=\E[\Phi(X_2)]$, from which the thesis follows.\ The proof in the case $x_1 < x_2 < 1$ is similar. 
To prove convexity of $\Phi$, we first show mid-convexity.\ 
Assume by contradiction that there exist $x_1, x_2 \in [0, +\infty)$ such that $\Phi((x_1+x_2)/2) > (\Phi(x_1) + \Phi(x_2))/2$.\ 
Then there exist $z \in [0,+\infty)$ and $\lambda \in (0,1)$ such that
\[
\lambda \Phi(z) + (1-\lambda) \Phi \left (\frac{x_1+x_2}{2} \right) > 1 > \lambda \Phi(z) + (1-\lambda) \frac{\Phi(x_1)+\Phi(x_2)}{2}.
\]
Take disjoint sets $A,B,C$ such that $P(A)=\lambda$, $P(B)=\frac{1-\lambda}{2}$ and $P(C)=\frac{1-\lambda}{2}$ and let
$X_1=z 1_A+ x_1 1_B+x_2 1_C$, $X_2=z 1_A+ x_1 1_C+x_2 1_B$, and
\[
X=z1_A+\frac{x_1+x_2}{2}1_B+\frac{x_1+x_2}{2}1_C= \frac{X_1+X_2}{2}.
\]
It holds that
\[
\E[\Phi(X_1)]=\E[\Phi(X_2)]=\lambda \Phi(z)+(1-\lambda)\of{\frac{\Phi(x_1)+\Phi(x_2)}{2}}<1,
\]
which implies $\trho(Y_1)<1$ and $\trho(Y_2)<1$.\ 
Similarly,
\[
\E[\Phi(X)]=\lambda \Phi(z)+(1-\lambda)\Phi\of{\frac{x_1+x_2}{2}}>1,
\]
which implies $\trho(X)>1$, contradicting convexity of $\trho$.\ Therefore, $\Phi$ is mid-convex.\
Since a nondecreasing mid-convex function is convex, the thesis follows.
\end{proof}

\section{Consistent scoring functions for Orlicz premia}\label{sec:scoring}

In this section, we show that Orlicz premia are elicitable and we study general families as well as specific examples of strictly consistent scoring functions.\

\subsection{Elicitability and strict consistency}
We start by recalling a few standard definitions adapted to the class of strictly positive return risk measures.
\begin{definition}[Elicitability and strictly consistent scoring functions]\label{def:eli}
A functional $\rho \colon L^\infty_{++} \to (0,+\infty)$ is elicitable if there exists a strictly consistent scoring function $S \colon (0,+\infty) \times (0,+\infty) \to [0,+\infty)$ satisfying
$S(x,y)\geq 0$ and $S(x,y)=0$ if and only if $x=y$ such that, for each $Y \in L^\infty_{++}$, it holds that
\begin{equation}\label{eq:eli}
\rho(Y) = \argmin_{x \in (0,+\infty)}\,  \E [S(x,Y)].
\end{equation}
The strictly consistent scoring function $S$ is said to be of the prediction error form if $S(x,y)=f(x-y)$ and of the relative error form if $S(x,y)=g(y/x)$, where $f,g$ are functions of a single variable.
\end{definition}

The following theorem provides a general, rich family of scoring functions that are strictly consistent with Orlicz premia.
\begin{theorem}\label{th:scores}
Let $\Phi$ be as in Definition~\ref{def:orl}, with $\Phi$ increasing. 
Let $h \colon (0,+\infty)\to (0,+\infty)$ be any integrable function. 
Then,
\begin{equation}
S_{\Phi}(x,y)=\int_{x}^y h(z)(\Phi(y/z)-1)\, \diff z
\label{eq:scscoring}
\end{equation}
is a strictly consistent scoring function for the Orlicz premium $H_{\Phi}$.
\end{theorem}
\begin{proof}
Let $x>H_{\Phi}(Y)$. 
We compute
\begin{align*}
&\E[S_{\Phi}(x,Y)]-\E[S_{\Phi}(H_{\Phi}(Y),Y)]\\
&=\E\sqb{\int_x^Yh(z)\of{\Phi(Y/z)-1}\,\diff z}-\E\sqb{\int_{H_{\Phi}(Y)}^Yh(z)\of{\Phi(Y/z)-1}\,\diff z}
\\
&=-\E\sqb{\int_{H_{\Phi}(Y)}^xh(z)\of{\Phi(Y/z)-1}\,\diff z}=-\int_{H_{\Phi}(Y)}^xh(z)\E[\Phi(Y/z)-1]\,\diff z> 0,
\end{align*}
where in the last line we have used Fubini's Theorem and item~(c) of Proposition~\ref{prop:orl}. 
The same argument shows that if $x<H_{\Phi}(Y)$, then again
\[
\E[S_{\Phi}(x,Y)]-\E[S_{\Phi}(H_{\Phi}(Y),Y)]=\int_{x}^{H_{\Phi}(Y)} h(z)\E[\Phi(Y/z)-1]\, \diff z > 0,
\]
so $S_{\Phi}$ is a strictly consistent scoring function for the Orlicz premium.
\end{proof}
\medskip 

Two particularly interesting cases arise by taking $h(z)=1/z$ and $h(z)=1/z^2$. 
In the first case, with the change of variable $y/z=e^t$, we obtain
\begin{align}\label{eq:orl-score}
S_{\Phi}(x,y)&=\int_x^y \frac{1}{z}(\Phi(y/z)-1)\,\diff z=\int_0^{\log(y/x)}(\Phi(e^t)-1)\, \diff t \notag \\
&=\varphi (\log (y/x)),
\end{align}
where
\begin{equation}\label{eq:varphi}
\varphi(t)= \int_0^t \left [ \Phi (e^s) -1 \right] \mathrm{d}s.
\end{equation}
In the second case, with the change of variable $y/z=t$, we obtain
\begin{equation}\label{eq:orl-score-2}
S_{\Phi}(x,y)=\int_1^{y/x} \frac{\Phi(t)-1}{y} \, \diff t.
\end{equation}

We now present a collection of examples of strictly consistent scoring functions for Orlicz premia, using Orlicz functions that are commonly adopted in the literature. 
As we will see, in some cases the general approach based on Theorem~\ref{th:scores} or the specific forms in equations \eqref{eq:orl-score} and \eqref{eq:orl-score-2} recover scoring functions that are already known in the literature, whereas in many other cases new families of scoring functions are obtained. 
In several cases the corresponding Orlicz premium $H_\Phi$ admits an explicit expression, whereas in some other cases it has to be computed numerically. 
All our examples, including some not discussed below, are summarized in Table~\ref{tab:1}.

\begin{example}[Mean]\label{ex:mean}
Let $\Phi(x)=x$. 
Then, $H_\Phi(Y)=\E[Y]$ and, from \eqref{eq:orl-score} and \eqref{eq:varphi}, we obtain
\begin{equation}\label{eq:S-mean}
S_{\Phi}(x,y)=\frac{y}{x} - \log \left (\frac{y}{x} \right) -1,
\end{equation}
which is an alternative scoring function for the mean. 
From the classical result of \cite{S71}, as recalled e.g., in Theorem~7 of \cite{G11}, the most general class of strictly consistent scoring functions for the mean belongs to the family of Bregman functions, of the form
\[
S(x,y) = \phi(y) - \phi(x) - \phi^\prime(x)(y-x),
\]
where $\phi$ is a convex function with subgradient $\phi^\prime$. 
The scoring function given in equation~\eqref{eq:S-mean} arises if $\phi(x) = -\log x$. 
This scoring function is known as the quasi-likelihood (QLIKE) scoring function in the econometrics literature, and it is of common use in assessing forecasts of nonnegative quantities such as volatility (see e.g.,\ \cite{P11} and the references therein).
\end{example}

\begin{example}[Expectiles]\label{ex:expectile_b}
Let $\Phi_q(x)=1+ q(x-1)^{+} -(1-q)(x-1)^{-}$ with $0 < q <1$, as in Example~\ref{ex:expectiles}.\
The corresponding Orlicz premium is the $q$-expectile and, from \eqref{eq:orl-score} and \eqref{eq:varphi}, the corresponding scoring function 
is given by
\begin{equation}
S_{\Phi}(x,y)=q(y/x-\log(y/x)-1)^{+}+(1-q)(y/x-\log(y/x)-1)^{-}.
\label{eq:S-expectile}
\end{equation}
The class of strictly consistent scoring functions for expectiles has been characterized in Theorem~10 of \cite{G11}, as an asymmetric extension of Bregman functions, defined by
\[
S(x,y) = \vert 1_{x \geq y} -q \vert \{ \phi(y) - \phi(x) - \phi^\prime(x)(y-x) \},
\]
and as in the case of the mean our scoring function~\eqref{eq:S-expectile} corresponds to the case $\phi(x) = -\log x$.
\end{example}

\begin{example}[$p$-norms]\label{ex:p-norms}
Let $\Phi(x)=x^p$ with $p\geq 1$. 
Then, the Orlicz premium $H_{\Phi}$ is given by
\[
H_{\Phi}(Y)=\Vert Y \Vert _p,
\]
and the corresponding scoring function from \eqref{eq:orl-score} and \eqref{eq:varphi} takes the form
\[
S_{\Phi}(x,y)=\frac{1}{p}\of{\frac{y^p}{x^p}-1}-\log\of{\frac{y}{x}}.
\]
In view of \eqref{eq:S-mean}, which arises as a special case when $p\equiv 1$, we will refer to this novel scoring function as the \textit{PQLIKE} scoring function.

When $0<p<1$, the resulting Orlicz premium is no longer a norm, but the scoring function is still valid. 
Taking a (suitably normalized and scaled) limit for $p\rightarrow 0$, such that $\Phi(x)=1+\log(x)$ arises, 
the Orlicz premium $H_{\Phi}$ is the logarithmic certainty equivalent given by 
\[
H_{\Phi}(Y)=\exp\of{\E[\log Y]}.
\]
The corresponding scoring function from \eqref{eq:orl-score} and \eqref{eq:varphi} takes the form
\[
S_{\Phi}(x,y)=\left(\log\of{\frac{y}{x}}\right)^2.
\]
\end{example}

\begin{example}[Mean-variance]\label{ex:MV}
Let $\Phi(x)=\lambda x+(1-\lambda)x^2$ with $0\leq \lambda \leq 1$ as in Section~5.1 of \cite{HG82}.
Then, the Orlicz premium $H_{\Phi}$ is given by
\[
H_{\Phi}(Y)=
\E[Y]
\left (\frac{\lambda}{2}+\sqrt{\left(\frac{\lambda}{2}-1\right)^2 +(1-\lambda) \frac{\Var[Y]}{\E^2[Y]}} \right),
\]
and the corresponding scoring function from \eqref{eq:orl-score} and \eqref{eq:varphi} is
\[
S_{\Phi}(x,y)=\frac{\lambda y}{x}+\frac{(1-\lambda)}{2}\frac{y^2}{x^2}-\log\left(\frac{y}{x}\right)-\frac{\lambda+1}{2}.
\]
\end{example}


\begin{example}\label{ex:exp}
Let $\Phi(x)=\frac{e^{\alpha x}-1}{e^{\alpha}-1}$ with $\alpha>0$ as in Section~5.2 of \cite{HG82}. 
The corresponding Orlicz premium $H_\Phi$ does not in general admit an explicit expression. 
It is the solution of the equation
$
f_Y(\alpha/H_\Phi)=e^\alpha,
$
where $f_Y(t)=\E[\exp(tY)]$ is the moment generating function of $Y$. 
For example, if $Y$ has a Gamma distribution with shape parameter $\theta>0$ and rate parameter $\gamma>0$, we have
\[
H_{\Phi}(Y)=  \frac{\alpha}{\gamma(1-\exp(-\alpha/\theta))}
=\frac{\alpha e^{\alpha/\theta}}{\theta(e^{\alpha/\theta}-1)}\E[Y].
\]
The corresponding scoring function from \eqref{eq:orl-score-2} is given by
\[
S_{\Phi}(x,y)=\frac{e^{\alpha}}{(e^{\alpha}-1)}\of{\frac{1}{\alpha y}\of{e^{\alpha\of{\frac{y}{x}-1}}-1}+\frac{1}{y}-\frac{1}{x}}.
\]
\end{example}

%

\landscape
\renewcommand\arraystretch{2.5}
\begin{table}
{\footnotesize\begin{tabular}{|c|c|c|}
\hline \hline 
Orlicz Function $\Phi(x)$ & Orlicz Premium $H_\Phi(Y)$ & Scoring Function $S_\Phi(x,y)$ \\
\hline
$x$ & $\E[Y]$ & $\frac{y}{x}-\log\of{\frac{y}{x}}-1$ \\
\hline
$\alpha+1_{\{x>1\}}$, $0<\alpha<1$ & $q_{\alpha}(Y)$ & $(1_{\{x\geq y\}}-\alpha)\log\of{\frac{x}{y}}$ \\
\hline
$1+q(x-1)^{+}-(1-q)(x-1)^{-}$, $0<q<1$ & $e_q(Y)$ & $q\of{\frac{y}{x}-\log\of{\frac{y}{x}}-1}^{+}+(1-q)\of{\frac{y}{x}-\log\of{\frac{y}{x}}-1}^{-}$ \\
\hline
$1+\log(x)$ & $\exp(\E[\log Y])$ & $\left(\log\of{\frac{y}{x}}\right)^2$ \\
\hline
$x^p$, $p\geq 1$ & $\Vert Y \Vert_p$ & $\frac{1}{p}\of{\frac{y^p}{x^p}-1}-\log\of{\frac{y}{x}}$ \\
\hline
$\lambda x^p+(1-\lambda)x^{2p}$, $0\leq\lambda\leq 1$, $p\geq 1$ & $\of{\frac{1}{2}\of{\lambda\E[Y^p]+\sqrt{\lambda^2\E[Y^p]^2+4(1-\lambda)\E[Y^{2p}]}}}^{\frac{1}{p}}$ & $\frac{1-\lambda}{2p}\of{\frac{y}{x}}^{2p}+\frac{\lambda}{p}\of{\frac{y}{x}}^p-\log\of{\frac{y}{x}}-\frac{\lambda+1}{2p}$ \\
\hline
$x\log(e-1+x)$ & \text{ no explicit form } & $\of{\frac{y}{x}+e-1}\of{\log\of{\frac{y}{x}+e-1}-1}-\log\of{\frac{y}{x}}$ \\
\hline
$x^p+1_{\{x>1\}}x^p\log(x)$, $p\geq 1$ & \text{ no explicit form } & $-\log\of{\frac{y}{x}}+\frac{1}{p}\of{\frac{y^p}{x^p}-1}+1_{\{y>x\}}\of{\frac{1}{p}\frac{y^p}{x^p}\log\of{\frac{y}{x}}-\frac{1}{p^2}\of{\frac{y^p}{x^p}-1}}$ \\
\hline
$\frac{e^{\alpha x}-1}{e^{\alpha}-1}$, $\alpha>0$ & $\frac{\alpha e^{\alpha/\theta}}{\theta(e^{\alpha/\theta}-1)}\E[Y]$ & $\frac{e^{\alpha}}{(e^{\alpha}-1)}\of{\frac{1}{\alpha y}\of{e^{\alpha\of{\frac{y}{x}-1}}-1}+\frac{1}{y}-\frac{1}{x}}$ \\
\hline
$xe^{\alpha(x^2-1)}$, $\alpha>0$ & $\sqrt{\frac{2}{\pi}}\of{\frac{1}{2}e^{-\alpha}+\sqrt{\frac{1}{4}e^{-2\alpha}+\pi\alpha}}\sigma$ & $\frac{1}{2\alpha y}\Big(e^{\alpha \of{\frac{y^2}{x^2}-1}}-1\Big)+\frac{1}{y}-\frac{1}{x}$ \\
\hline
$\frac{e^x-x-1}{e-2}$ & \text{ no explicit form } & $\frac{1}{e-2}\of{\frac{e^{\frac{y}{x}}-\frac{1}{2}}{y}+\frac{(2-2e)x-y}{2x^2}}$ \\
\hline \hline
\end{tabular}
\caption{\footnotesize{Examples of Orlicz functions with the corresponding Orlicz premia and the corresponding strictly consistent scoring functions. 
Eqn.~\eqref{eq:orl-score} is used to obtain the scoring functions of the first eight examples. 
For the remaining examples, \eqref{eq:orl-score-2} is used.
Although $\Phi(x)=\alpha+1_{\{x>1\}}$ does not satisfy the conditions of Theorem~\ref{th:scores}, it still gives rise to a strictly consistent scoring function.
When $\Phi(x)=\frac{e^{\alpha x}-1}{e^{\alpha}-1}$, the corresponding Orlicz premium does not in general admit an explicit expression. 
The expression given in the table corresponds to the case in which $Y$ follows a Gamma distribution with shape parameter $\theta>0$ and rate parameter $\gamma>0$.
When $\Phi(x)=xe^{\alpha(x^2-1)}$, the corresponding Orlicz premium also does not in general admit an explicit expression. 
The given expression corresponds to the case in which $Y=\abs{Z}$, where $Z$ follows a normal distribution with mean $0$ and standard deviation $\sigma$.}}
\label{tab:1}}
\end{table}
\endlandscape

\subsection{Mixture representations}\label{sec:mixture}

It is clear that a scoring function of the form~\eqref{eq:scscoring} depends on the choice of the function $h$. 
Hence, the ranking of competing forecasts may depend on this choice, in particular in finite samples and under model misspecification (see e.g., \cite{P20}). 
To remedy the dependence of the ranking on the specific choice of the strictly consistent scoring function, \cite{EGJK16} develop a method to compare forecasts simultaneously with respect to a class of strictly consistent scoring functions by considering so-called Murphy diagrams. 
This method relies on the availability of a mixture representation of the strictly consistent scoring functions under consideration, 
in terms of elementary scoring functions depending on a low-dimensional parameter. 
Mixture representations for the class of strictly consistent scoring functions for quantiles and expectiles have been given in \cite{EGJK16}. 
A mixture representation of strictly consistent scoring functions for the triplet of Range Value-at-Risk and its two associated Value-at-Risks has been given in \cite{FZ21}. 
In the following theorem, we provide such a mixture representation for our new family of scoring functions in~\eqref{eq:scscoring}.

\begin{theorem}\label{th:mixrep}
Any strictly consistent scoring function for the Orlicz premium $H_{\Phi}$ of the form~\eqref{eq:scscoring} admits a representation of the form
\begin{equation}\label{eq:mixrep}
S_{\Phi}(x,y)=\int_{0}^{+\infty}S_z(x,y)\,\diff H(z),
\end{equation}
for a positive measure $H$, where $S_z(x,y)=\abs{\Phi\of{\frac{y}{z}}-1}1_{\cb{x\leq z < y}\cup \cb{y\leq z < x}}$.  
Conversely, for any choice of the positive measure $H$, we obtain a strictly consistent scoring function of the form~\eqref{eq:scscoring} for the Orlicz premium $H_{\Phi}$.
\end{theorem}
\begin{proof}
By using~\eqref{eq:scscoring}, we have
\begin{align*}
S_{\Phi}(x,y)&=\int_x^y h(z)\left(\Phi(y/z)-1\right)\,\diff z \\
&=\int_x^y\left(\Phi(y/z)-1\right)\,\diff H(z) \\
&=\int_0^{+\infty}S_z(x,y)\,\diff H(z),
\end{align*}
where $\diff H(z)=h(z)\,\diff z$ and $S_z(x,y)=\abs{\Phi\of{\frac{y}{z}}-1}1_{\cb{x\leq z < y}\cup \cb{y\leq z < x}}$. 
Note that, since the function $h$ in~\eqref{eq:scscoring} is strictly positive, the Riemann integral $\int_0^{+\infty}S_z(x,y)\,\diff H(z)$ is well-defined.
\end{proof}
\medskip 
	
As a corollary, we provide a mixture representation for the scoring functions of $p$-norms.
\begin{corollary}\label{cor:mixreppnorm}
Any strictly consistent scoring function of the form~\eqref{eq:scscoring} with $\Phi(x)=x^p$, $p\geq 1$, admits a representation of the form:
\begin{equation}\label{eq:mixreppnorm}
S_{\Phi}(x,y)=\int_{0}^{+\infty}S^p_z(x,y)\,\diff H(z),
\end{equation}
for a positive measure $H$, where $S^p_z(x,y)=\abs{y^p-z^p}1_{\cb{x\leq z < y}\cup \cb{y\leq z < x}}$. 
\end{corollary}
\begin{proof}
From~\Cref{th:mixrep} with $\Phi(x)=x^p$, we obtain
\begin{align*}
S(x,y)&=\int_0^{+\infty} \abs{\frac{y^p}{z^p}-1}1_{\cb{x\leq z < y}\cup \cb{y\leq z < x}}\,\diff H(z) \\
&=\int_0^{+\infty}\abs{y^p-z^p}1_{\cb{x\leq z < y}\cup \cb{y\leq z < x}}\,\diff \widetilde{H}(z),
\end{align*}
where $\diff\widetilde{H}(z):=\frac{1}{z^p}\,\diff H(z)$.
\end{proof}
\medskip 

We conduct two simulation experiments to illustrate how one can use Theorem~\ref{th:mixrep} to rank competing forecasts. 
In particular, we will generate Murphy diagrams for logarithmic certainty equivalents, $p$-norms and expectiles by using the corresponding elementary scoring functions.
\begin{example}\label{ex:p-normMurphy}
We first suppose that the true distribution of the outcome variable $Y$ is given by $\log(Y)|\mu\ \sim\mathcal{N}(\mu,\sigma_{Y}^{2})$ where $\mu\sim\mathcal{N}(0,\sigma_{\mu}^{2})$.
We take $\sigma_{Y}=\sigma_{\mu}=0.2$. 
We consider four different forecasters, who will be referred to as perfect, unconditional, unfocused and sign-reversed, similar to \cite{GBR07} and \cite{EGJK16}, suitably modified to the current setting. 
The perfect forecaster issues the true distribution of the outcome $Y$ as predictive distribution. 
Therefore, his/her point forecasts of the logarithmic certainty equivalent (LCE) and $p$-norm are $\exp(\mu)$ and $\exp(\mu+\sigma_{Y}^{2}p/2)$ with $\sigma_{Y}=0.2$. 
The unconditional forecaster does not have knowledge of $\mu$ and issues the unconditional distribution of $\log(Y)$ as predictive distribution: $\mathcal{N}(0,\sigma_{\mu}^{2}+\sigma_{Y}^{2})$. 
Therefore, his/her point forecasts of the LCE and $p$-norm are $\exp(0)$ and $\exp((\sigma_{\mu}^{2}+\sigma_{Y}^{2})p/2)$ with $\sigma_{Y}=\sigma_{\mu}=0.2$.  
The remaining two forecasters, unfocused and sign-reversed, have knowledge of $\mu$, but their predictive distributions fail to be ideal. 
The unfocused forecaster issues a mixture distribution as predictive distribution of $\log(Y)$, involving an independent random variable $\tau$ that takes the values $0.2$ and $-0.2$ with probability $1/2$, leading to $\tfrac{1}{2}(\mathcal{N}(\mu,\sigma_{Y}^{2})+\mathcal{N}(\mu+\tau,\sigma_{Y}^{2}))$ yielding $\exp(\mu+\tau/2)$ and $\exp(\mu+\tau/2+\sigma_{Y}^{2}p/4)$ with $\sigma_{Y}=0.2$ as the corresponding forecasts of the LCE and $p$-norm. 
The sign-reversed forecaster issues a predictive distribution of $\log(Y)$ with the sign of $\mu$ flipped: $\mathcal{N}(-\mu,\sigma_{Y}^{2})$. 
Therefore, his/her point forecasts are $\exp(-\mu)$ and $\exp(-\mu+\sigma_{Y}^{2}p/2)$. 
The point forecasts generated by the four predictive distributions are summarized in Table~\ref{tab:forecastp-norm}. 
Using $10\mathord{,}000$ simulations of sample size $1\mathord{,}000$ each, we obtain the Murphy diagrams displayed in Figure~\ref{fig:p-norm} for the LCE and $p$-norm with $p=1,2,3$.
As can be seen in Figure~\ref{fig:p-norm}, the perfect forecaster dominates the other forecasters for the LCE and all $p$-norms considered, as expected.
Although not clearly visible, the expected scores for the other three forecasters intersect in all four cases, such that none of these forecasters dominates the other. 
\begin{table}[h!]
{\small \centering
\begin{tabular}{|c|c|c|}
\hline
Forecaster & Predictive distribution of $\log(Y)$ & Point forecast of $p$-norm \\
\hline \hline
Perfect & $\mathcal{N}(\mu,\sigma_{Y}^{2})$ & $\exp(\mu+\sigma_{Y}^{2}p/2)$ \\
\hline
Unconditional & $\mathcal{N}(0,\sigma_{\mu}^{2}+\sigma_{Y}^{2})$ & $\exp((\sigma_{\mu}^{2}+\sigma_{Y}^{2})p/2)$ \\
\hline
Unfocused & $\tfrac{1}{2}(\mathcal{N}(\mu,\sigma_{Y}^{2})+\mathcal{N}(\mu+\tau,\sigma_{Y}^{2}))$ & $\exp(\mu+\tau/2+\sigma_{Y}^{2}p/4)$ \\
\hline
Sign-reversed & $\mathcal{N}(-\mu,\sigma_{Y}^{2})$ & $\exp(-\mu+\sigma_{Y}^{2}p/2)$ \\
\hline \hline
\end{tabular}
\caption{Predictive distributions and point forecasts. 
The point forecasts for the LCE arise by taking $p\equiv 0$.}
\label{tab:forecastp-norm}}
\end{table}
		
\begin{figure}[h!]
\centering
\begin{subfigure}{0.49\textwidth}
\includegraphics[width=\textwidth]{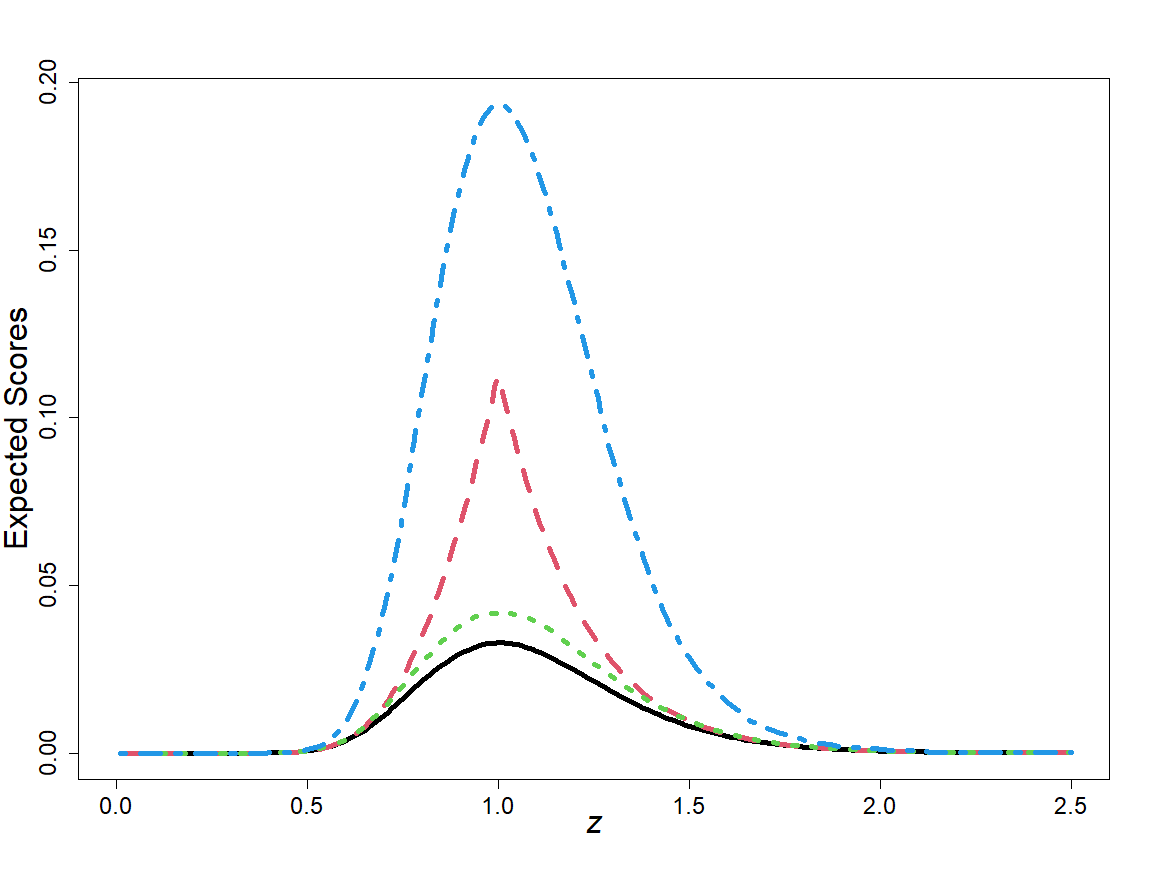}
\caption{Expected scores for\\ the LCE}
\end{subfigure}
\begin{subfigure}{0.49\textwidth}
\includegraphics[width=\textwidth]{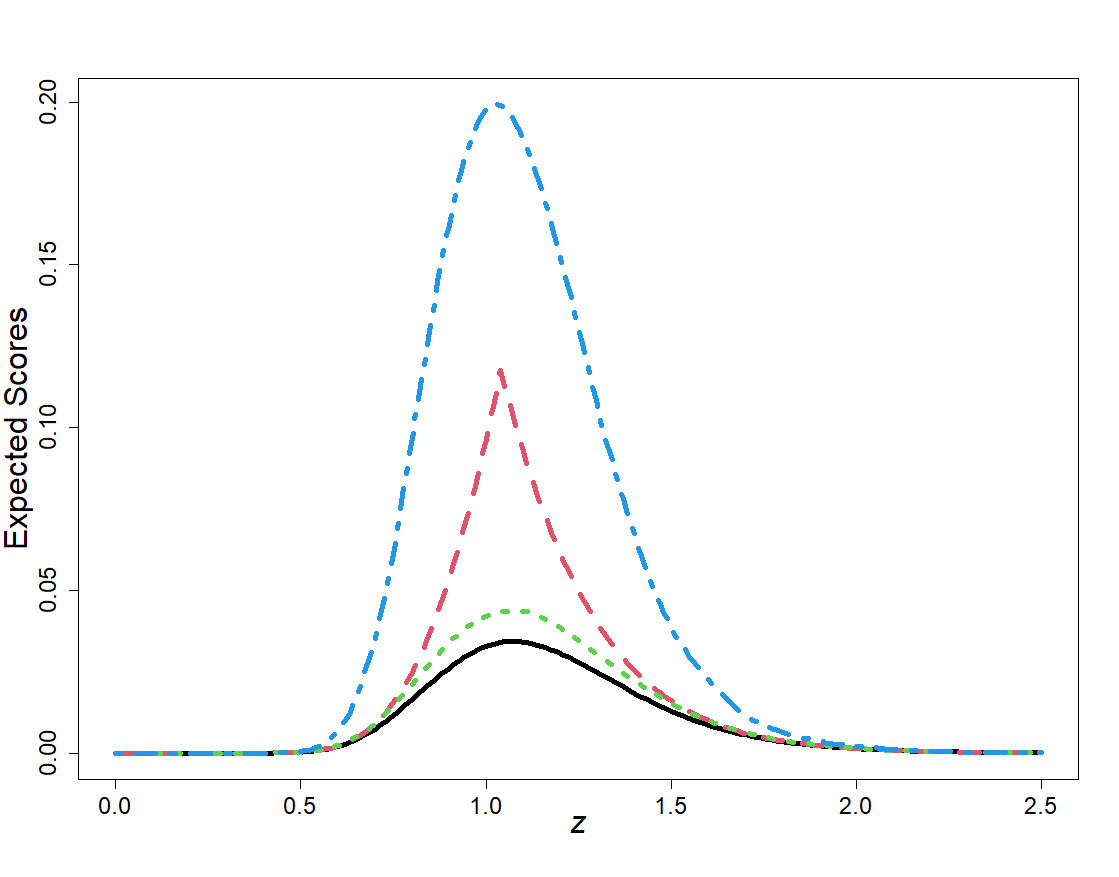}
\caption{Expected scores for\\ $p=1$}
\end{subfigure}
\hfill
\begin{subfigure}{0.49\textwidth}
\includegraphics[width=\textwidth]{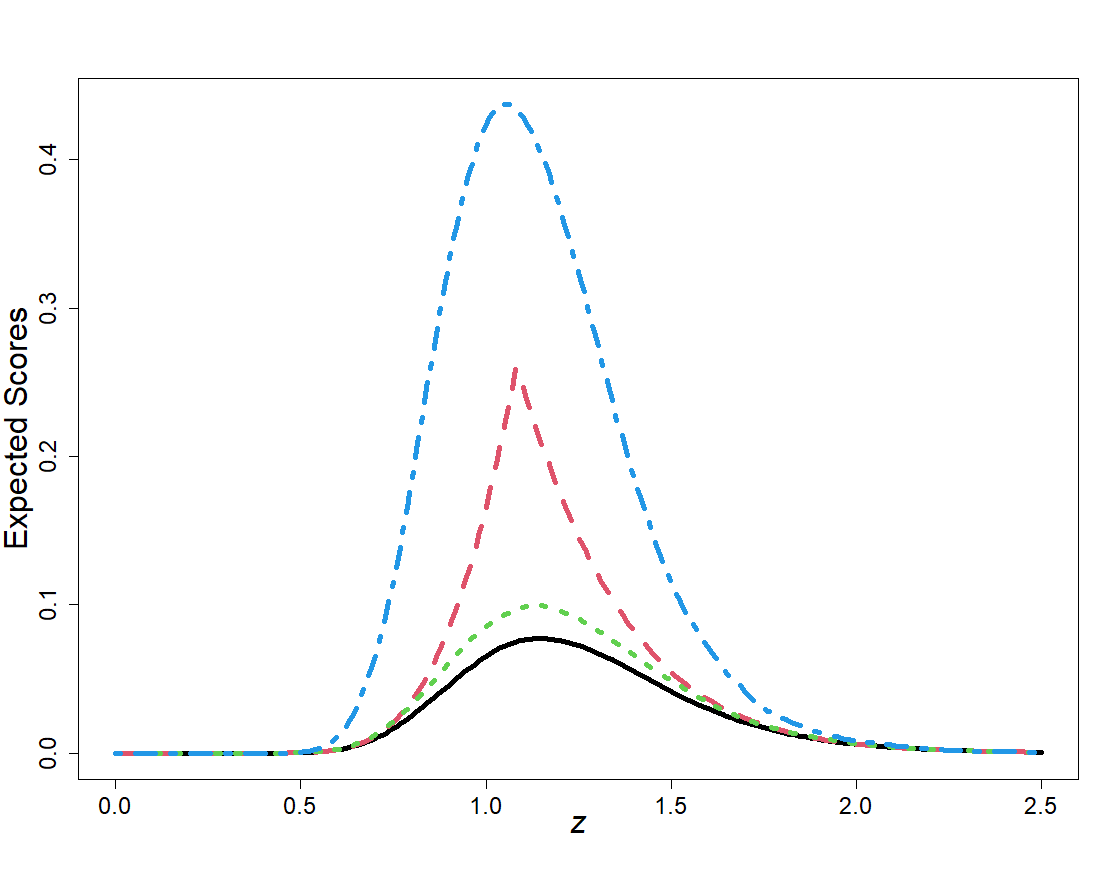}
\caption{Expected scores for \\$p=2$}
\end{subfigure}
\hfill
\begin{subfigure}{0.49\textwidth}
\includegraphics[width=\textwidth]{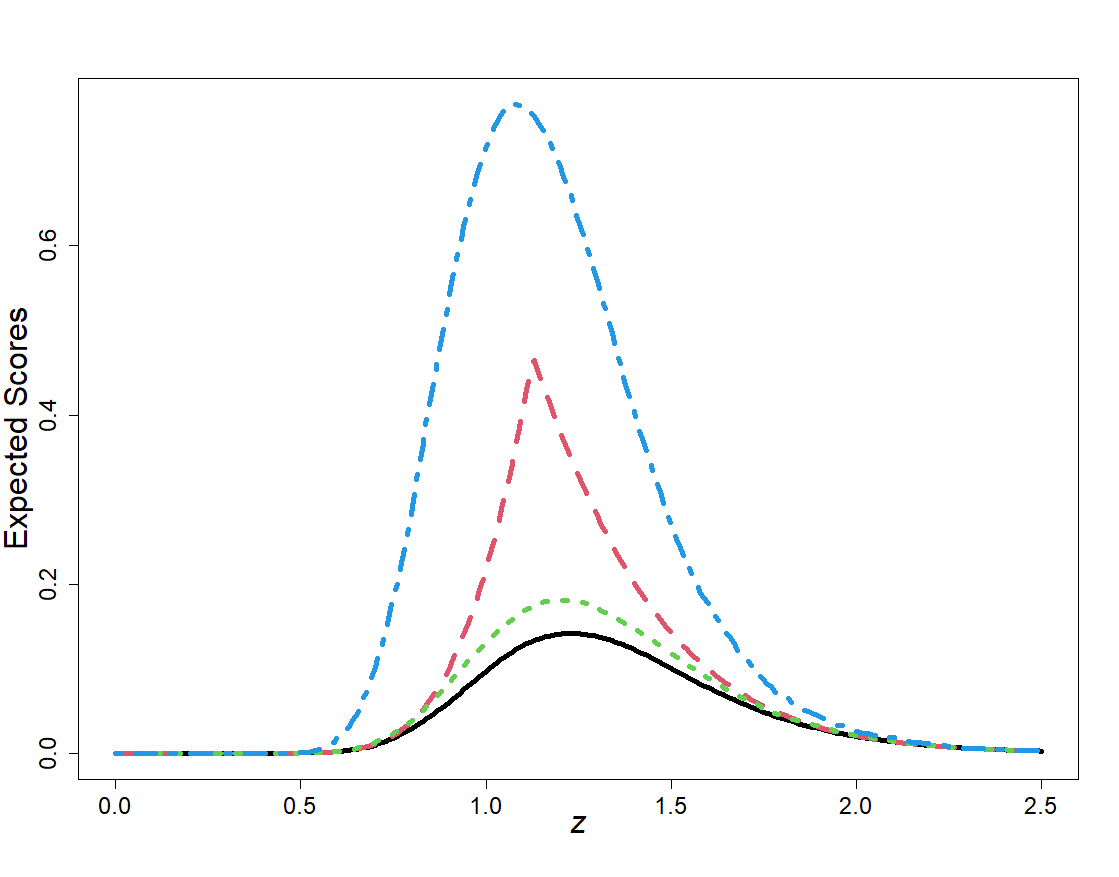}
\caption{Expected scores for \\$p=3$}
\end{subfigure}
\caption{Murphy diagrams of the four forecasters for the LCE and different $p$-norms (perfect: black; unconditional: red; unfocused: green; sign-reversed: blue)}
\label{fig:p-norm}
\end{figure}
\end{example}

\begin{example}
For $q$-expectiles, from Theorem~\ref{th:mixrep} with $\Phi_q(x)=1+ q(x-1)^{+} -(1-q)(x-1)^{-}$, $0 < q <1$, we have 
\begin{align*}
S(x,y)&=\int_0^{+\infty}\abs{q\of{\frac{y}{z}-1}^{+}-(1-q)\of{\frac{y}{z}-1}^{-}}1_{\cb{x\leq z <y}\cup \cb{y\leq z < x}}\diff H(z)\\
&=\int_0^{+\infty}\abs{q(y-z)^{+}-(1-q)(y-z)^{-}}1_{\cb{x\leq z <y}\cup \cb{y\leq z < x}}\diff \widetilde{H}(z),
\end{align*}
where $\diff \widetilde{H}(z):=\frac{1}{z}\diff H(z)$. 
Therefore, the elementary scoring function for expectiles can be expressed as 
\begin{equation}
S_z^q(x,y)=q(y-z)1_{\cb{x\leq z <y}}+(1-q)(z-y)1_{\cb{y\leq z <x}}.
\label{eq:elmscoringexpectile}
\end{equation}

Suppose now that the true distribution of the outcome variable $Y$ is given by $Y|\lambda\sim\exp(\lambda)$ where $\log(\lambda) \sim \mathcal{N}(0,\sigma^{2}_{\lambda})$.
We take $\sigma_{\lambda}=0.2$.
We consider three different forecasters: perfect, unfocused and mean-reversed, similar to Example~\ref{ex:p-normMurphy}. 
The perfect forecaster issues the true distribution of $Y$ as predictive distribution. 
The unfocused forecaster issues $\exp(\tau\lambda)$ as 
predictive distribution, involving an independent random variable $\tau$ that takes the values $5/4$ and $4/5$ each with probability $1/2$.  
The mean-reversed forecaster issues a predictive distribution with the mean reversed: $\exp(1/\lambda)$. 
The point forecasts generated by the three predictive distributions are displayed in Table~\ref{tab:forecastq-expectile}. 
Using $10\mathord{,}000$ simulations of sample size $1\mathord{,}000$ each, we obtain the Murphy diagrams displayed in Figure~\ref{fig:q-expectile} for $q=0.5,0.7,0.9,0.95$.
As we see from the figure, the perfect forecaster dominates the other forecasters, as expected. 
There is no ordering relationship between the unfocused and mean-reversed forecasters, because their expected scores intersect.
\begin{table}[h!]
{\small \centering
\begin{tabular}{|c|c|c|}
\hline
Forecaster & Predictive distribution of $Y$ & Point forecast of $q$-expectile \\
\hline \hline
Perfect & $\exp(\lambda)$& $\frac{1}{\lambda}\of{1+W\big(\frac{2q-1}{(1-q)e}\big)}$ \\
\hline
Unfocused & $\exp(\tau\lambda)$ & $\frac{1}{\tau\lambda}\of{1+W\big(\frac{2q-1}{(1-q)e}\big)}$ \\
\hline
Mean-Reversed & $\exp(1/\lambda)$ & $\lambda\of{1+W\big(\frac{2q-1}{(1-q)e}\big)}$ \\
\hline \hline
\end{tabular}
\caption{Predictive distributions and point forecasts. 
$W$ denotes the Lambert function. 
}
\label{tab:forecastq-expectile}}
\end{table}

\begin{figure}[h!]
\centering
\begin{subfigure}{0.49\textwidth}
\includegraphics[width=\textwidth]{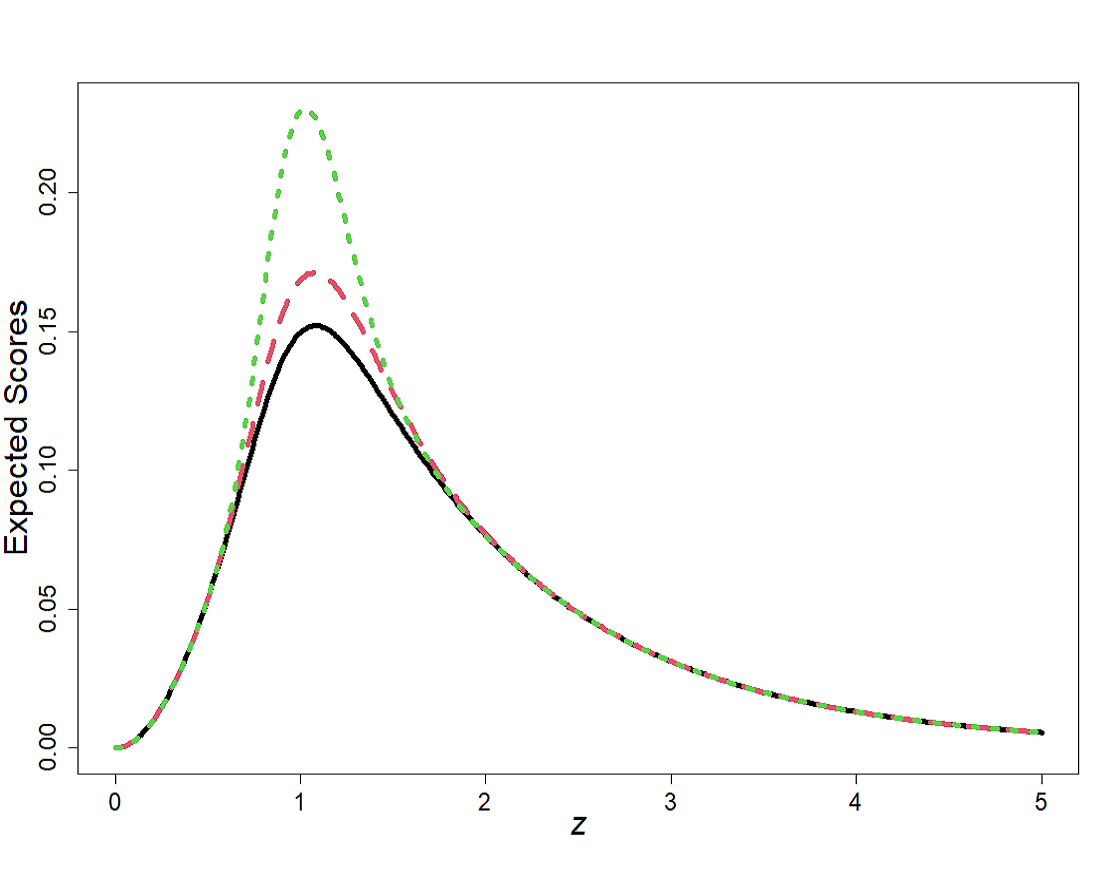}
\caption{Expected scores for \\$q=0.5$}
\end{subfigure}
\hfill
\begin{subfigure}{0.49\textwidth}
\includegraphics[width=\textwidth]{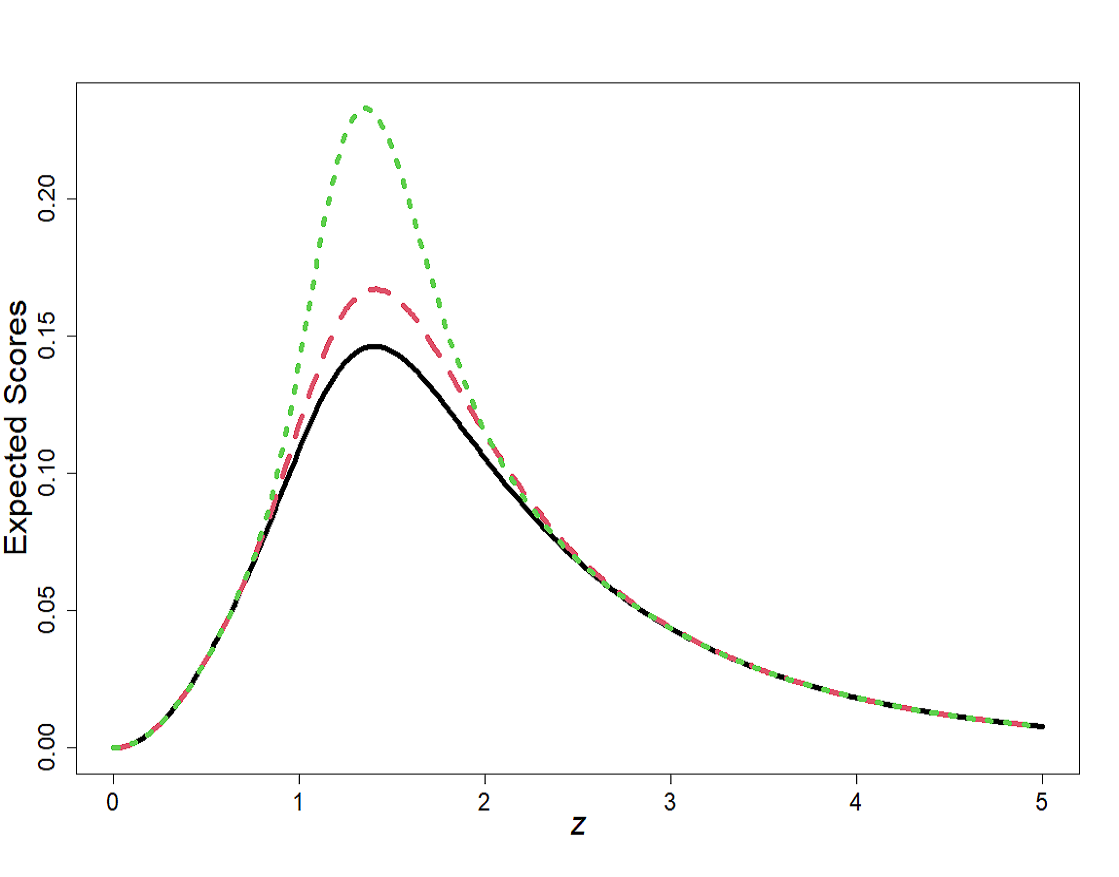}
\caption{Expected scores for \\$q=0.7$}
\end{subfigure}
\hfill
\begin{subfigure}{0.49\textwidth}
\includegraphics[width=\textwidth]{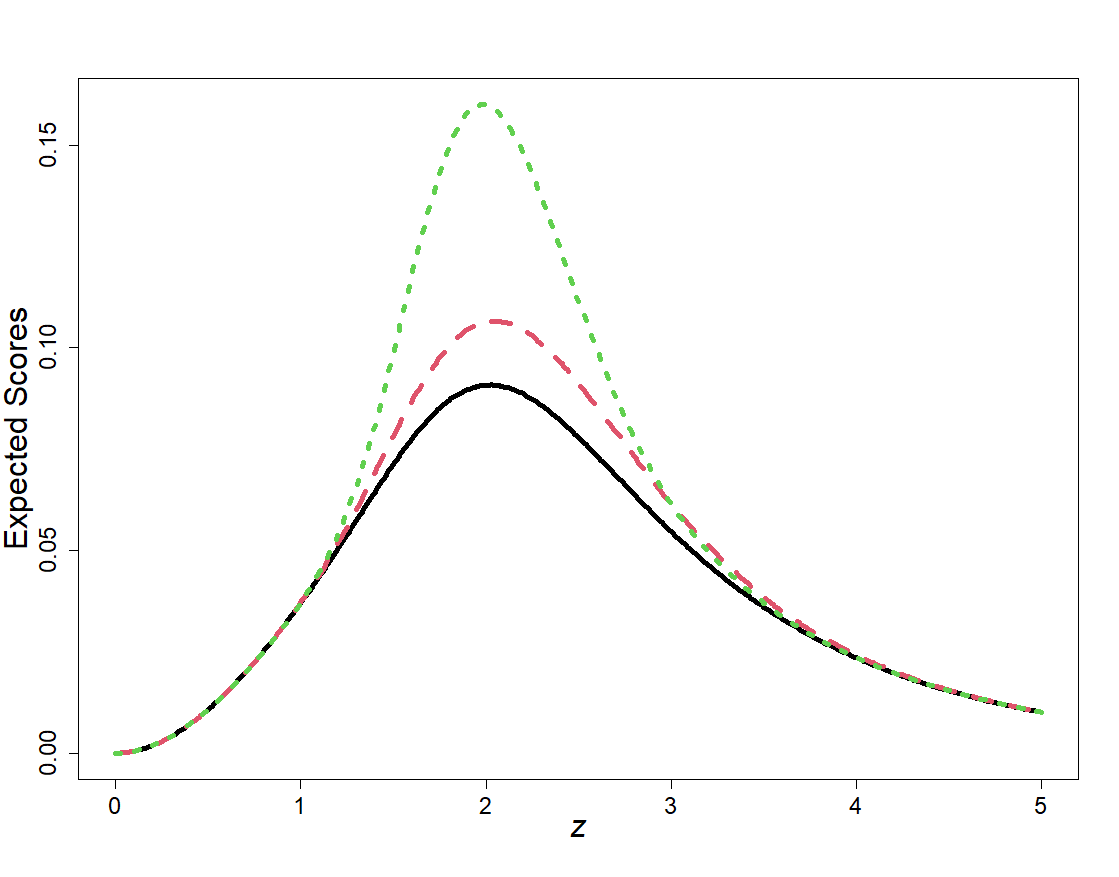}
\caption{Expected scores for \\$q=0.9$}
\end{subfigure}
\hfill
\begin{subfigure}{0.49\textwidth}
\includegraphics[width=\textwidth]{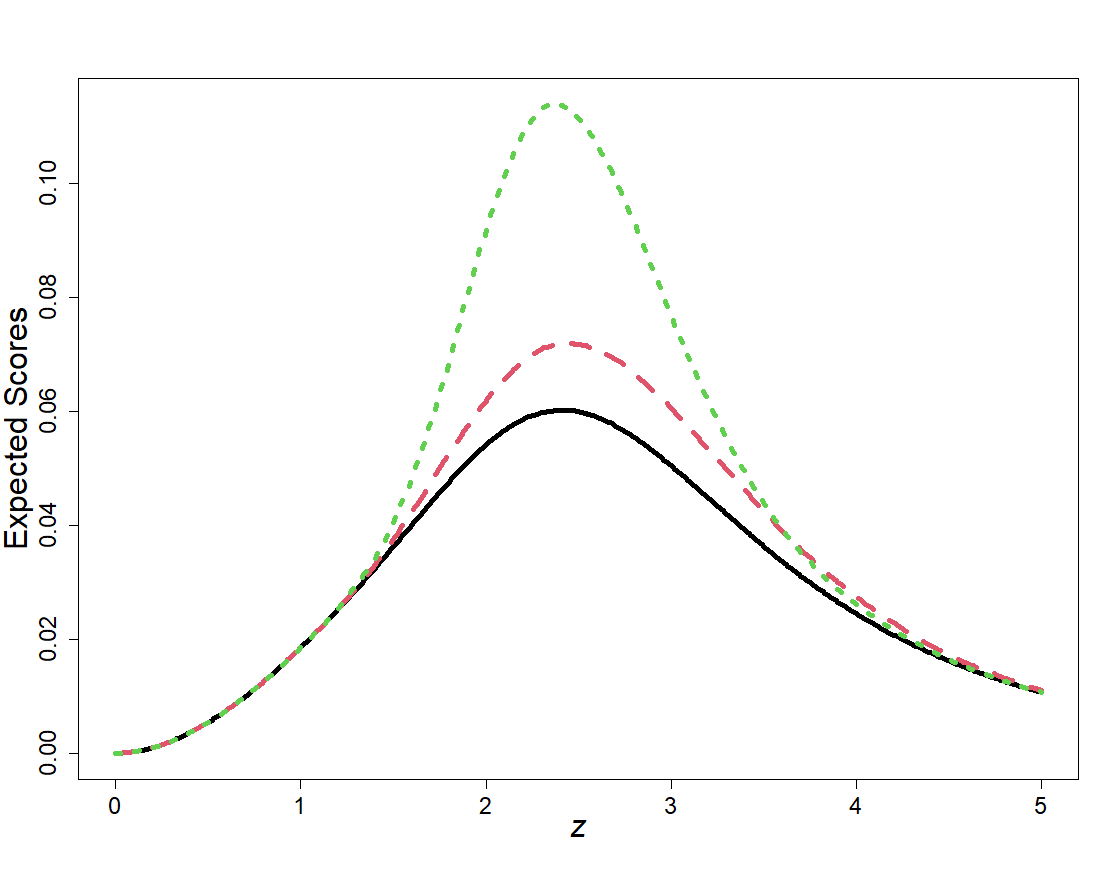}
\caption{Expected scores for \\$q=0.95$}
\end{subfigure}
\caption{Murphy diagrams of the three forecasters for different $q$-expectiles (perfect: black; unfocused: green; mean-reversed: red)}
\label{fig:q-expectile}
\end{figure}
\end{example}
 
\newpage\ \newpage

\section{Appendix}
\begin{lemma}\label{lemma:AA-GA}
Let $f \colon (0,+\infty) \to \R$ be nondecreasing and convex. 
Then, $f$ is GA-convex.
\end{lemma}
\begin{proof}
For $x,y>0$ and $\lambda \in (0,1)$ from the AM-GM inequality it holds that $\lambda x +(1-\lambda)y\geq x^{\lambda}y^{1-\lambda}$. 
Since $f$ is nondecreasing and convex it follows that
\[
f(x^{\lambda}y^{1-\lambda})\leq f(\lambda x +(1-\lambda)y)\leq \lambda f(x)+(1-\lambda)f(y),
\]
which gives the thesis.
\end{proof}\medskip

\begin{proof}[Proof of Lemma~\ref{lemma:ident}]
Since $\trho$ is law invariant and positively homogeneous with $\trho(1)=1$, it follows that $\trho(\delta_y) =y$, for each $y>0$. 
From identifiability, it follows that
\begin{align*}
I(x,y) &= 0 \iff x=y, \\
I(x,y) &> 0 \iff x<y.
\end{align*}

For each $0 < y_1 < x < y_2$, define

\begin{align*}
\bar{p} &= \frac{I(x, y_2)}{I(x, y_2) - I(x, y_1)} \in (0,1), \\
\bar{F} &= \bar{p} \delta_{y_1} + (1-\bar{p}) \delta_{y_2}.
\end{align*}
Since
\begin{align*}
\int I(x,y) \,\mathrm{d}\bar{F}(y) &=\bar{p} I(x,y_1)+ (1-\bar{p}) I(x, y_2)  \\
&= \frac{I(x, y_1)I(x, y_2)}{I(x, y_2) - I(x, y_1)} + \frac{- I(x, y_2)I(x, y_1)}{I(x, y_2) - I(x, y_1)}  = 0,
\end{align*}
it follows that $\trho(\bar{F}) = x$. 
From the positive homogeneity of $\trho$, it follows that for each $\lambda >0$,
\[
\trho ((1-\bar{p}) \delta_{\lambda y_1} + \bar{p} \delta_{\lambda y_2}) = \lambda x,
\]
and from identifiability
\[
\int I(\lambda x,y) \,\mathrm{d} \left[ (1-\bar{p}) \delta_{\lambda y_1} + \bar{p} \delta_{\lambda y_2} \right] =0,
\]
which gives
\[
\frac{- I(x, y_2)I(\lambda x, \lambda y_1)}{I(x, y_1) - I(x, y_2)} + \frac{I(x, y_1)I(\lambda x, \lambda y_2)}{I(x, y_1) - I(x, y_2)} = 0,
\]
so we can conclude that
\begin{equation*}
\frac{I(x, y_1)}{I(x,y_2)} = \frac{I( \lambda x, \lambda y_1)}{I( \lambda x, \lambda y_2)},
\end{equation*}
and letting $\lambda =1/x$ we find that
\begin{equation} \label{eq:main}
\frac{I(x, y_1)}{I(x, y_2)} = \frac{I( 1, y_1/x)}{I( 1 , y_2/x)},
\end{equation}
for each $0 < y_1 < x < y_2$. 

We now want to prove that
\begin{equation}\label{eq:thesis}
I(x,y) = g(y/x) \cdot h(x),
\end{equation}
with $h(x) >0$, from which the thesis follows immediately. We consider two cases.

If $y >x$, we set $y_1 =x/2$ and $y_2 = y$ in \eqref{eq:main}, obtaining
\[
\frac{I(x, x/2)}{I(x, y)} = \frac{I(1, 1/2)}{I(1 , y/x)},
\]
which gives
\begin{equation}\label{eq:case1}
I(x, y) = {I(1, y/x)} \cdot \frac{I(x, x/2)}{I(1, 1/2)}.
\end{equation}

If instead $y<x$, we set $y_1 =y$ and $y_2 = 2x$ in \eqref{eq:main}, obtaining
\[
\frac{I(x, y)}{I(x, 2x)} = \frac{I(1, y/x)}{I(1, 2)},
\]
which gives
\begin{equation}\label{eq:case2}
I(x, y) = I(1, y/x) \cdot \frac{I(x, 2x)}{I(1, 2)}.
\end{equation}
Notice also that from~\eqref{eq:main} it follows that
\[
\frac{I(x, x/2)}{I(1, 1/2)} = \frac{I(x, 2x)}{I(1, 2)},
\]
so combining \eqref{eq:case1} and \eqref{eq:case2} it follows that \eqref{eq:thesis} is satisfied with
\begin{align*}
g(t) &= I(1,t), \\
h(x) &= \frac{I(x, x/2)}{I(1, 1/2)},
\end{align*}
from which the thesis follows.
\end{proof}


\setstretch{0.92}

\end{document}